%% file: sanos.tex
\documentclass[twoside]{article}
\DeclareRobustCommand{\removeblock}[1]{
\relax
}
\usepackage[utf8]{inputenc}
\usepackage[T1]{fontenc}
\usepackage{anyfontsize}
\AtBeginDocument{\fontsize{10.5}{12.6}\selectfont}
\usepackage{amsmath, amssymb, amsfonts, amsthm}
\usepackage{mathtools}
\usepackage{enumitem}
\usepackage[numbers]{natbib}
\usepackage{graphicx}
\usepackage{booktabs}
\usepackage{siunitx}
\usepackage{lmodern} 
\usepackage{microtype}
\usepackage[table, dvipsnames]{xcolor}
\usepackage[most]{tcolorbox}
\usepackage[normalem]{ulem}
\definecolor{oxfordblue}{rgb}{0.0, 0.13, 0.28}
\definecolor{linkblue}{rgb}{0.0, 0.20, 0.40}
\usepackage{hyperref} 
\usepackage{nicematrix}
\hypersetup{
    colorlinks = true,
    linkcolor = linkblue,
    anchorcolor = linkblue,
    citecolor = linkblue,
    filecolor = linkblue,
    urlcolor = linkblue
}

\definecolor{limegreen}{rgb}{0.2, 0.8, 0.2}
\definecolor{darkgreen}{rgb}{0.0, 0.5, 0.0}
\definecolor{darkred}{rgb}{0.6, 0.0, 0.0}
\definecolor{bondiblue}{rgb}{0.0, 0.58, 0.71}
\definecolor{britishracinggreen}{rgb}{0.0, 0.26, 0.15}
\definecolor{faintgray}{RGB}{245,245,245}
\definecolor{faintborder}{RGB}{230,230,230}
\definecolor{lightblack}{gray}{0.4}
\definecolor{pinegreen}{rgb}{0.0, 0.47, 0.44}
\input{mathcommands}

\input{Comments}

\DeclareRobustCommand{\R}{\mathbb{R}}

\theoremstyle{plain}
\newtheorem{theorem}{Theorem}[section]        
\newtheorem{proposition}[theorem]{Proposition}

\theoremstyle{definition}

\newtheorem{definition}[theorem]{Definition}

\theoremstyle{remark}
\newtheorem{remark}[theorem]{Remark}

\DeclareRobustCommand{\supstack}[2]{\stackrel{\mbox{{\scriptsize ${#1}$}}}{ {#2} }}     

\newcounter{question}
\newtcolorbox[auto counter, use counter=question]{question}[1][]{
  enhanced,
  colback=faintgray,
  colframe=faintborder,
  boxrule=0.2pt,
  arc=2mm,
  title=\textcolor{lightblack}{\textbf{Question~\thequestion}},
  fonttitle=\bfseries,
  before upper={\centering\itshape},
  after title={\vspace{0.5ex}},
  boxsep=4pt,
  left=6pt,
  right=6pt,
  top=4pt,
  bottom=4pt,
  #1
}

\usepackage[a4paper, top=2.5cm, bottom=2.5cm, left=2.7cm, right=2.7cm]{geometry}

\title{
    \hrule height 4pt
    \vspace{0.25 in}
    {\Huge\bfseries SANOS} \\ [0.4em] 
    {\Large Smooth strictly Arbitrage-free Non-parametric Option Surfaces}
    \vspace{0.25in} 
    \hrule height 1pt
}

\author{%
  \hspace{-.18cm} 
  \textbf{Hans~Buehler}\textsuperscript{1}%
  \footnote{All authors listed in alphabetical order.}%
  \,\,\footnote{Corresponding author.}
  \quad
  \textbf{Blanka~Horvath}\textsuperscript{1,2} 
  \quad
  \textbf{Anastasis~Kratsios}\textsuperscript{3,4,5} 
  \\[0.5em] 
  \hspace{-.18cm} 
  \textbf{Yannick~Limmer}\textsuperscript{6}  
  \quad 
  \textbf{Raeid~Saqur}\textsuperscript{1,5}  \\[1em]
  \parbox{0.9\textwidth}{\centering\small
    \textsuperscript{1}Mathematical Institute, University of Oxford \\
    \textsuperscript{2}Oxford-Man Institute for Quantitative Finance, University of Oxford \\
    \textsuperscript{3}Department of Mathematics, McMaster University \\
    \textsuperscript{4} The Ennio De Giorgi Mathematical Research Centre, Scuola Normale Superiore di Pisa\\
    \textsuperscript{5}Vector Institute \\
    \textsuperscript{6}DRW
  }
}

\date{May 22nd, 2026{\small \\First version January 14th, 2026}}

\begin{document}


 \maketitle
\begin{abstract}
    
   We present a simple, numerically efficient non-parametric method to construct representations of option price surfaces
    $\hat C(T,K)$ which are both smooth and strictly arbitrage-free across time and strike. 
    The method can be viewed as a smooth generalization of the widely-known linear interpolation scheme, 
    and retains the simplicity and transparency of that baseline. Our method is able to provide smooth fits
    to full option surfaces from 0DTE to the farthest expiry in seconds.

Call option prices for arbitrary strikes $K$ and some expiry $T_j$ are represented as convex combinations of Black--Scholes call payoffs anchored at strikes $K^i_j$ with variances~$V_j$ 
with weights~$q_j$\vspace{-0.3cm}
$$
\vspace{-0.1cm}\hat C(T_j,K)=\sum_{i=0}^{N_j} q_j^i\,\mathrm{Call}(K^i_j,K,V_j).\vspace{-0.05cm}
$$ 
which are free of arbitrage across time and strike. This is extended to the full model $\hat C(T,K)$ to include arbitrary non-quoted expiries. The~$q_j$'s are shown  to have a natural interpretation as discrete space transition densities.

Calibration of the model to observed market quotes on different expiries and strikes than those supporting the model is formulated as a linear program, allowing bid--ask spreads to be incorporated directly via linear penalties or inequalities, and delivering materially lower computational cost than most of the currently available implied-volatility surface fitting routines.

As a further contribution, we derive an equivalent parameterization of the proposed surface in terms of strictly positive “discrete local volatility” \cite{DLV} variables. This yields, to our knowledge, the first construction of smooth, strictly arbitrage-free option price surfaces while requiring only trivial parameter constraints (positivity). We illustrate the approach using S\&P\,500 index options.
\end{abstract}

\section{Introduction}

This article tackles the problem of fitting a discrete set of observed European option bid-ask prices efficiently and develops  a \textit{smooth, strictly arbitrage-free, non-parametric surface} representation of option prices to model prices for a full continuous range of strikes and maturities beyond the observed market quotes. As long as bid/ask spreads of option prices
are not zero, and if the market is free of arbitrage, we can always find a version of our model which fits the market.
We also present a version of our model which requires only positivity of the input parameters,
making our model the first known smooth strictly arbitrage-free model whose parameters are subject to only trivial constraints. \\

Figure~\ref{fig:initialexample} makes the strength of the calibration immediately apparent. In a way, these plots speak more than a thousand words: Across the full range of strikes and maturities shown, the model curves closely track the observed S\&P\,500 option data, with deviations that are difficult to discern at the scale of the plots and no evident systematic bias in any particular region of the surface. Notably, this high-quality fit is achieved with negligible runtime—solving a single global linear program over the full cross-section of options across all expiries completes essentially instantaneously.

\textsl{The data used for illustration in this article are end of day option data sourced from Option Metrics IvyDB.\footnote{\url{https://optionmetrics.com/}}. We selected OTM options with at least 100 lots are traded on the day. We also exclude strikes of options whose $\mathrm{Vega}/\sqrt{T}$ is less than 0.1\%. Past this, we also only consider expiries with
20~active options. }

\begin{figure}[H]
    \centering
    \includegraphics[width=0.9\linewidth]{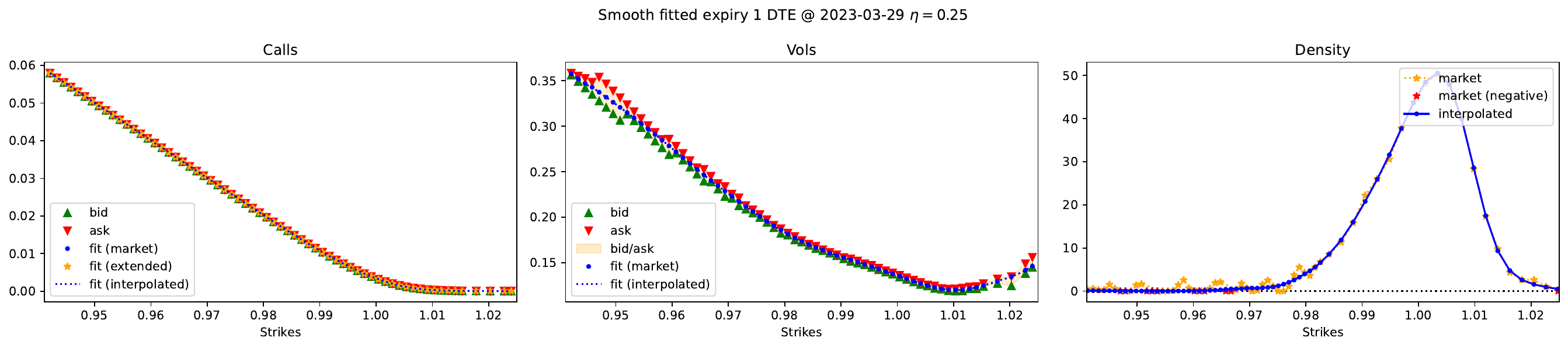}\\
    \includegraphics[width=0.9\linewidth]{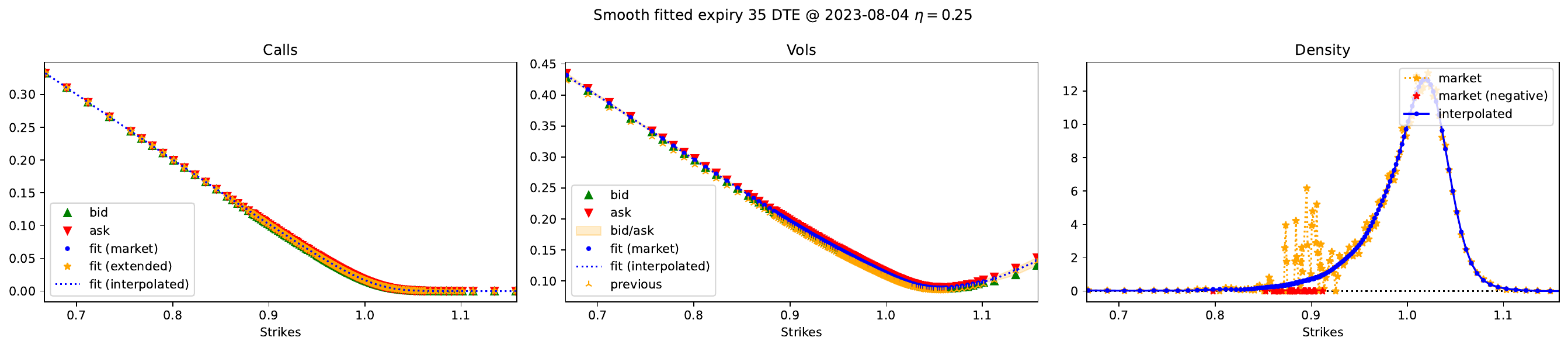}\\
    \includegraphics[width=0.9\linewidth]{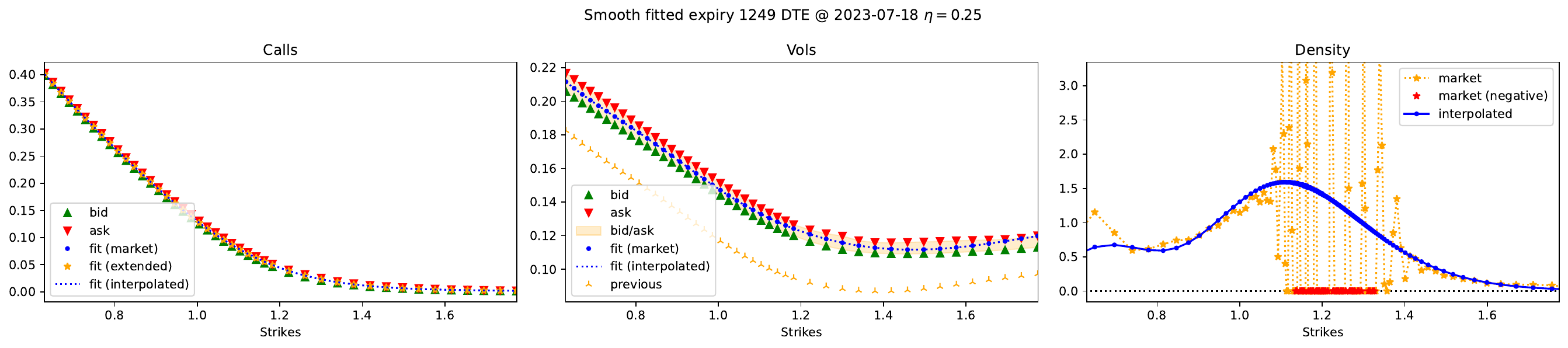}
    \caption{Fit of our model to the market, illustrated for a few DTE's on several days in 2023 using IvyDB OptionMetrics data. The thin bars in the call price and volatility
    graphs illustrate fit errors, expressed relative to the relevant spread on the right hand axis. (We do see some fitting errors in these charts because of the smoothness
    we impose; see section~\ref{sec:prod_model}.) The red dots below the zero line in the density
    graphs represent cases where the market-mid prices have arbitrage.
    }
    \label{fig:initialexample}
\end{figure}

A convenient feature of our model is that it is strictly arbitrage-free and does not require any of the compromises between fitting and avoiding arbitrage that other currently available approaches rely on: such as Lagrange multipliers, penalties, or reliance on ``observations in experiments''. Moreover,
its implementation is simple and fast enough to be used in high speed trading environments or
as part of machine learning models
trained using gradient descent or other heavy numerical methods.

Our proposed production variant boils down to writing
\begin{equation}\label{eq:intromodel}
    \hat C(T_j,K) := \sum_{i=0}^{N_j} q^i_j\, \mathrm{Call}\big(K^i_j,K,V_j \big) \ .
\end{equation}
in terms of martingale densities~$q_1,\ldots,q_M$ for expiries~$T_1<\cdots<T_M$
defined for strikes $K^1_j<\cdots<K^{N_j}_j$ and log-normal variances $0<V_1<\cdots<V_M$. Here, $\mathrm{Call}(s,k,v)$ is the Black-Scholes call price function 
for spot~$s$, strike~$k$ and variance~$v$. We find the densities~$q_1,\ldots,q_M$ to match the market
bid/ask using linear programming.

The method enforces absence of arbitrage in both time and strike space for the marginals~$\hat C(T_k,\cdot)$. 
In order to retain the linear structure of our approach with respect to~$q$ we postulate that the model price for an arbitrary expiry~$T\in (T_{j-1},T_j]$
and strike~$K$ can be written in terms of some increasing weight function $\alpha_j(T) \in [0,1]$
with $\alpha_j(T_{j-1})=0$ and $\alpha_j(T_j)=1$ as
\begin{equation}\label{eq:intro_full_C}
    \hat C(T, K) := \left\{
            \begin{array}{rl}
               {\color{pinegreen}\alpha_j(T)} & \sum_{i=1}^{N_j} {\color{pinegreen} q_j^i} \mathrm{Call}\!\left({\color{pinegreen} K^i_j },  K, {\color{pinegreen} V(T) } \right)  \\
        + {\color{pinegreen}(1 - \alpha_j(T))} & \sum_{i=1}^{N_{j-1}} {\color{pinegreen}q_{j-1}^i } \mathrm{Call}\!\left( 
        {\color{pinegreen} K^i_{j-1} }, K, {\color{pinegreen}V(T)} \right) 
        \end{array}
        \right.
\end{equation}
where we use the interpolated variance
$$
    V(T) := \alpha_j(T) V_j + (1-\alpha_j(T)) V_{j-1} \ .
$$
Note that this preserves smoothness in $K$ for all $T>0$.
An example for $\alpha_j(T)\in[0,1]$ is the linear
$$
    \alpha_j(T) := \frac{ ( T\wedge T_j - T_{j-1} )^+ }{ T_j - T_{j-1} } \ .
$$
In  practice a~$C^1$ monotone spline is preferable. (In the above $N_0=1$, $q_0^1:=1$ and $K^1_1:=1$.)

\subsubsection*{Fitting to Market Data}

Assume now our market has expiries $0=\tau_0<\cdots<\tau_m$ with market strikes $k^r_\ell$ for $\ell=1,\ldots,m$
and $r=1,\ldots,n_\ell$.
The model price~$\hat C(\tau_\ell, k^r_\ell)$ for each $k^r_\ell$ with $\tau_\ell\in (T_{j-1}, T_j]$ given in~\eqref{eq:intro_full_C} is linear in the variables~$q$, hence we can fit $q$ to market bid/ask prices by linear/convex programming provided our objective is linear/convex. This is discussed in 
section~\ref{sec:PROD} on implementation on page~\pageref{sec:PROD}.

Our model fits market data surprisingly well and is by construction strictly free of arbitrage and smooth. Our model has a simple interpretation as a discrete jump process multiplied
by an underlying martingale process.

In addition to our production model, we also present extensions to other underlying processes than Black-Scholes, and generalize our approach to a theoretically more
appealing model  which allows specification of different variance curves~$V^i_j$
per strike. While this model is theoretically more attractive its calibration is much heavier
and it is therefore perhaps of more academic interest.

\subsubsection*{Accessible Parametrization}

Considering the increased interest in generative models for option prices we also present a parametrization
of our model in terms of a parameter set which is subject only to positivity constraints. To this end we express the densities~$q_1,\ldots,q_M$
in our formula (which are subject to linear but still complicated constraints) with densities generated by ``Discrete Local Volatilities'' (DLVs) as introduced in~\cite{DLV} (which only need to be non-negative).
These DLVs together with the forward volatilities to generate the backbone variances only need be positive to generate an arbitrage-free smooth call option
price surface.

We note that a model in that space, combined with~\cite{removedrift, buehler2022deephedginglearningremove} will yield not only arbitrage-free  snapshots, but
also discrete time dynamics which are free of dynamic arbitrage. Such a model is subject to further research.

\subsubsection*{Implementation}

We provide a specific implementation strategy for both models in section~\ref{sec:implementation} .

\subsubsection*{Related Work}

There is an extensive body of work focusing on interpolating option implied volatilities with the aim of 
ensuring absence if strike and calendar arbitrage, starting as early as~\cite{estimatingimpliedMayhew}
and more recently
\cite{impliedvolsplinesBliss, fengler2009, deschâtres2026}. However,
none of the proposed models is able to ensure strict absence of arbitrage. 

Our approach is a generalization of the linear fit
to market prices first discussed in~\cite{DLV}.
This was more recently expanded in~\cite{cohen2020detecting, cohenreisingerwang2020}
to fit a linear arbitrage-free model between bid/ask prices. 
Both works expand on the similar, but much heavier
ideas of~\cite{buehler2006expensive} where a linear model is fitted to the market which also
reprices forward start options. The underlying linear model had -- in effect -- been already
been discussed in the seminal~\cite{Andersen1997TheEO, AndersenHugeDML},
albeit not in the context of fitting to the market, but
in the context of developing a consistent pricing operator once an arbitrage-free set of option prices is obtained.

In fact, linear interpolation already addresses many of the challenges of fitting a model to an arbitrary surface of observed option prices.
When applied to a dense grid it can be used to regularize local volatilities as discussed in~\cite{DLV}.
However, linear interpolation suffers from being ``the most expensive'' interpolation for call and put prices as
was pointed out in~\cite{buehler2006expensive}. This means that implied volatilities of linearly interpolated
option prices is too large as is illustrated in figure~\ref{fig:linexp}. This makes linear interpolation particularly
unsuitable if the option surface is used in automated trading applications such as Deep Hedging~\cite{DH} 
or other methods which use the prices of the surface to make trading decisions.
A~smooth interpolation scheme
avoids ``too expensive'' prices.

Our underlying model formulation is similar to~\cite{elalaoui2023}
who uses a time-homogeneous grid and assumes absence of arbitrage in the market prices
to construct the market density upon which he then attempts to train per-strike variances
via a non-linear least-squares optimisation. Using different variances per strike means that 
the resulting model no longer avoids term-structure arbitrage, a topic which we comment on below.

Seemingly similar are \textsl{Gaussian mixture models} which have been covered in the literature in the past, chiefly~\cite{mixedBrigo02}
and their earlier~\cite{mixedBrigo00}. Their primary focus is the construction of a local volatility
diffusion which mimics the marginal distributions generated by Gaussian mixture models.
The central difference, apart from the difference in focus, of these works to ours is that in \cite{mixedBrigo00,mixedBrigo02} the mixing is performed with fixed weights over time over a number of processes which all start at the same point; while our solution works by jumps from one underlying stock price to another:
The latter difference can be stated more explicitly, in our work as
    $$
        \hat C(T_j,K) := \sum_{i=0}^{N_j} q^i_j\, \mathrm{Call}\big({\color{pinegreen} K^i_j},K,V_j\big) 
    $$
where we used $S_0=1$ for spot and variances~$V_j$. Mixture models on the other hand implement
    $$
        \sum_{i=0}^{N_j} q^i_j\, \mathrm{Call}\big({\color{pinegreen} 1},K,V_j\big) \ . 
    $$
The effect is that we can impose much stronger skew via the weighting functions~$q_j$, at the cost
of needing to limit the weights~$q_j$ to densities with unit mean. As a result, we are
able to fit any grid of arbitrage-free bid/ask prices, while Gaussian mixture
models typically fail to do so.\\

We focus here on non-parametric models which aim to fit a market, rather than on (low-) parametric
representations of the option surface which express a view on the market
through the lens of intuitive and statistically meaningful parameters.
Usually such models exhibit arbitrage -- with the exception of the SSVI models
presented in the seminal~\cite{gatheral2014arbitrage}. This is a great tool to express a view on the market
through a low-parametric model lens.
However, as the
the authors point out, the most accessible parametrization
does not always fit well to observed option markets which means they need to resort to more involved
numerical methods to fit their low-parametric model to observed market data while retaining
strict absence of arbitrage.\\

\begin{table}[H]
    \centering
    \begin{tabular}{lcccc}
    \toprule
    Method
    & Smooth
    & Arbitrage-free
    & Parametric
    & Complexity \\
    \midrule
    Linear interpolation
    & $\times$
    & \checkmark
    & No
    & Linear \\
    SSVI
    & \checkmark
    & \checkmark
    & Yes
    & High \\
    This paper
    & \checkmark
    & \checkmark
    & No
    & Linear \\
    \bottomrule
    \end{tabular}
    \caption{Comparison of option surface construction methods.}
    \label{tab:placeholder}
\end{table}

\subsubsection*{Classic Calibration Models}

 We omitted in this brief summary all methods which rely on calibrating a heavier dynamic 
 model such as Heston~\cite{heston1993closed}, Bates~\cite{NBERw4596}, or various Levy models e.g.~\cite{carr2017local, CGMY2002}.
 In each of these
 cases model parameters are trained by minimizing the distance between market and model prices using numerical methods.
 The models used for this purpose in practise tend to be affine, in which case the Fast-Fourier pricing methods
 pioneered in~\cite{carrmadan1999} can be applied during calibration.  
 
 Further advances were made in~\cite{TangentLevy} where the authors show how to fit Levy models to an already arbitrage-free
 market. 
 Such methods are very promising
 but do not yet achieve the same quality of fit with the same speed as SANOS. 
 
\subsubsection*{American Options}
    In this article we focus on European options which are quoted for futures and indices. Options on single underlyings
    and ETFs, in contrast, are typically American options. 
    This is subject to further research.

\subsubsection*{Structure of the Article}

We start in Section~\ref{sec:review} with a review of the concept of absence of arbitrage and its relation to linear interpolation option prices. This section also serves
to settle notation and ideas.
Section~\ref{sec:smootharbfree} then covers the main two proposed call price surface models where we also present some illustrative numerical results. We also reflect on generalization
and representation as dynamic processes. Section~\ref{sec:implementation} covers practical implementation.

\section{Arbitrage-Free Call Prices}\label{sec:review}
We start by fixing some notions and conventions that we will frequently need throughout the paper. 
For this we first briefly recall  the Black--Scholes formula in the form we will frequently use in the paper:  With $s>0$ denoting spot,  $k\geq 0$ strike and $v\ge 0$ (total) variance\footnote{Many standard settings
denote variance $v$ as $v=\sigma^2 T$ with maturity $T\ge 0$ and volatility $\sigma\ge 0$.
}, the formula for European call options is given in terms of the
standard normal distribution function $\mathcal{N}$ as
\begin{equation}\label{eq:bsc}
    \mathrm{Call}(s,k,v) := s\,\mathcal{N}(d_+) - k\,\mathcal{N}(d_-) \qquad\text{with}\qquad
    d_\pm = \frac{-\ln(k/s)\pm \frac12 v}{\sqrt{v}} \, .
\end{equation}
 In this article we discuss construction of smooth arbitrage-free European option price surfaces for an equity denoted by~$S$. The following assumptions hold throughout the paper unless otherwise stated.
We will assume that the dynamics of the equity $S$ can be written in the form 
$S_T = F_T Z_T$ where~$F$ is the forward curve\footnote{
    The assumption $S_T=F_T Z_T$ implies that dividends are proportional; alternative cash dividend treatments and 
handling of simple credit risk which still allow conversion back to ``pure'' option prices are covered in~\cite{BuehlerVolAndDivs}.
} 
and where~$Z$ is a martingale with unit expectation. We also assume that the equity does not
default and does not otherwise reach zero\footnote{
We therefore exclude models which can diffuse into zero such as the CEV model given as $dS_t=\sqrt{S_t}\,dW_t$.
}. 
The deterministic discount factors are given by $\mathrm{DF}_T$. The key point is that in this setting 
if $\mathcal{C}(T,\mathcal{K})\equiv \mathrm{DF}_T\, \E[ (S_T-\mathcal{K})_+ ]$ is a surface of call prices with expiries $T$ and cash strikes $\mathcal{K}$, then we can convert 
these into \textit{``pure''} call prices on the martingale $Z$ by setting $\E[ (Z_T-K)_+] \equiv C(T,K) := \mathcal{C}(T,F_T K)/(\mathrm{DF}_T F_T)$.

A~similar transformation is applied to obtain pure put prices, which we then convert
into pure call prices by put-call-parity. 
With these preparations, we can henceforth focus or attention on ``pure'' call prices, i.e.~those written on the underlying martingale process:\\

\begin{definition}[Arbitrage Free Call Price Surfaces]\label{def:arbfreecall}

We call a function $C:[0,\infty)^2 \rightarrow [0,\infty)$ an \emph{arbitrage-free call price function} (or \emph{``surface''}) if 
    there is a probability measure $\P$  and a positive $\P$-martingale $(Z_T)_{T\in [0,\infty)}$
    with $Z_0=1$ a.s.,
    such that\footnote{This is an application of
    the ``First Fundamental Theorem of Asset Pricing''~\cite{delbaen2006general}
    which posits the equivalence to absence of arbitrage if and only if 
    a martingale~$Z$ exists.}
    $$
        C(T,K) = \E_\P\left[ (Z_T-K)^+ \right] \ .
    $$
    In this case we call~$Z$ the \emph{martingale representation} of~$C$.
\end{definition}

\noindent
We will by convention consider derivatives of call prices in $K$ as
right hand derivatives%
\footnote{
    i.e.~$\partial_K C(T,K) := \partial^+_K C(T,K) :=
    \lim_{\epsilon\downarrow 0} \frac1\epsilon ( C(T,K+\epsilon)-C(T,K) )$;
    see the proof of theorem~\ref{th:contarbfree} in the appendix on 
    page~\pageref{sec:appendix_absence} why this is the natural definition
    when working with distributions.}.
Next, we state a theorem that formulates a list of shape requirements on a surface, which guarantee to yield an arbitrage free call price surface in the sense of definition \ref{def:arbfreecall}.
\begin{theorem}\label{th:contarbfree}
    A~call price surface $C:[0,\infty)^2 \rightarrow [0,\infty)$ \emph{arbitrage-free} if and only if
    \begin{enumerate}
        \item The market has unit expectation: $C(T,0)\equiv \E[Z_T]=1$ for all $T$;\label{it:caf_C0_1} 
        
        \item zero is unattainable (i.e.~no atom at $0$): $\partial_K C(T,0)\equiv -1$;\label{it:caf_C1m1}
        \item call prices ultimately reach zero: $\lim_{K\uparrow \infty} C(T,K)=0$,\label{it:caf_limC}
        
        \item call prices are convex: $\partial_{KK} C(T,K) \geq 0$; and\label{it:caf_cvx}
        
        \item call prices are increasing in time: $\partial_T C(T,K)\geq 0$.\label{it:caf_dt}
    \end{enumerate}
    As a consequence we also have:
    \begin{enumerate}\setcounter{enumi}{5}
        \item $C(T,K) \geq (1-K)_+$ for any $K,T$; as a direct consequence of 1. \& 5. and
        \item $\partial_K C(T,K) \in [-1,0)$
        for any $K,T$.
        \label{it:derivpos}
    \end{enumerate}
\end{theorem}
\begin{remark}
Several of the above conditions are well-known in the context of related literature. 
In several related settings, strict absence of arbitrage statements are not always fully complete: the boundary conditions~\textit{\ref{it:caf_C0_1}.} and~\textit{\ref{it:caf_limC}.}~are frequently omitted, and/or negativity of the derivative~\textit{\ref{it:derivpos}.} is not required.

Our set of conditions for theorem \ref{th:contarbfree} are minimal equivalent conditions for the
existence of a positive true martingale in sense of definition~\ref{def:arbfreecall}.
A~full proof is delegated to the appendix page~\pageref{sec:appendix_absence}ff. 
\end{remark}

\begin{remark}
To ensure positivity, we have imposed condition~\textit{\ref{it:caf_C1m1}.}
rather than the commonly used  weaker condition $\partial_K C(T,0)\geq -1$.\footnote{
            For fixed $T$ and $K>0$, 
            we have, $ \partial_K C(T,K)
                =\E\!\left[\partial_K(X_T-K)^+\right]
                =-\P(X_T>K).$
            Taking the right-limit as $K\downarrow 0$ yields $\partial_K C(T,0)
                =-\P(X_T>0)
                =-(1-\P(X_T=0)).$}
 This requirement is justified as we consider an asset reaching  zero
outside a~default event undesirable, and point to~\cite{BuehlerVolAndDivs} on how to incorporate ``sudden'' default risk in our framework.
\end{remark}

\begin{definition}[Smooth call price function]
    We call a call price function \emph{smooth} (in strike direction) if $\partial_{KK} C(T,K) > 0$ for all $K>0$.
\end{definition}

\subsection{Discrete Martingales}

In this section we recall results from~\cite{DLV} which we also use to
fix notation and ideas for the remainder of the paper.  The purpose of this section is to anchor our
model in its non-smooth limit of linear interpolation between model prices.

Assume therefore for now we are given expiries $0=T_0<T_1<\cdots<T_M$ each with $N_j$ strikes
 $0=K^0_j<K^1_j<\cdots<1<\cdots<K^{N_j}_j$ and call prices $C^0_j,\ldots,C^{N_j}_j$
where $C_j^i$ denotes the price of the call expiring in $T_j$ with strike $K^i_j$. We also assume w.l.g.~(see remark~\ref{rem:extendstrikes}) that the boundary strikes $K^1_j\equiv K^{\min}$ and $K^{N_j}_j\equiv K^{\max}$ are
all the same.
Define for $i=0, \ldots, N_j-1$
\begin{equation}\label{def:dC}
    dC^i_j := \frac{ C^{i+1}_j - C_j^i }{ K^{i+1}_j - K^i_j }  
\end{equation}
and set $dC_j^{N_j}:=0$. We wish to understand when this set of market data is ``arbitrage-free'':
\begin{theorem}\label{th:noarb}
    A market of call prices $C_j^i$ is \emph{arbitrage-free} in the sense that 
    there exists a positive (discrete state) martingale~$X$ with unit expectation which satisfies
    $C^i_j=\E_\P[(X_{T_j}-K^i_j)^+]$
    if
    \begin{enumerate}
    		\item the market has unit expectation: $C_j^{0}=1$;
    		\item zero is unattainable: $dC_j^{0}=-1$ which with the above and \eqref{def:dC} implies $C^1_j=1-K^{\min}$;\footnote{
Equation \eqref{def:dC} (for $i=0$), and  condition $dC^0_j=-1$ gives
$
    (C^{1}_j - C^{0}_j)/(K^{1}_j-K_j^{0})=-1
$,
hence 
$C^{1}_j = C^{0}_j - (K_j^{1}-K_j^{0}).$
Using  $C^{0}_j=1$ (condition~1.), and the fact that $K_j^{0}=0$, gives
$C^{1}_j = 1 - (K_j^{1}-0)=1-K_j^{1}.$
}
        \item call prices ultimately reach zero: $C_j^{N_j}=0$;
        \item call prices are convex: $ dC^i_j \leq dC^{i+1}_j$; and\label{item:distimedc}
        \item call prices are increasing in time when linearly interpolated in the following sense:
          for $K^i_{j+1} \in [ K_j^\ell, K_j^{\ell+1})$ we have
          \begin{equation}\label{eq:arbfreeinT}              
                 C_j^\ell \frac{K_j^{\ell+1} - K_{j+1}^i}{K_j^{\ell+1} - K_j^\ell}
                 +
                 C_j^{\ell+1} \frac{K_{j+1}^i - K_j^{\ell} }{K_j^{\ell+1} - K_j^\ell}                 
                 \leq
                 C_{j+1}^i \ .
          \end{equation}
          If the strikes are homogeneous across expiries this reduces to
          the simpler
          $$
            dC^i_j \leq dC^i_{j+1} \ .
          $$

    \end{enumerate}
    As a consequence we also have
    \begin{enumerate}\setcounter{enumi}{5}
        \item $C_j^i \geq (1-K^i_j)^+$ and
        \item $dC_j^i\in [-1,0]$.
    \end{enumerate}
\end{theorem}
\noindent
\begin{proof}
See~\cite{DLV} theorems~3.1 and~3.2 and subsequent statements there for a constructive proof of the existence of a martingale $X$ with these marginal distributions in the  general case of non-homogeneous grids.
 It is also shown there that the corresponding martingale~$X$ only has mass in $K^1_j,\ldots,K^{N_j}_j$.
\end{proof}

\begin{remark}\label{rem:extendstrikes}
A common practice\footnote{
        A~viable strategy is as follows: first, to find~$K^{\max}$ let $K^*_j$ be the strike where the intrinsic value intersects with the line from $(K^3_j,C^3_j)$ through $(K^2_j,C^2_j)$, 
        i.e.~it satisfies $C^2_j+dC^2_j (K^*_j-K^3_j) = 1-K^*_j$.
         It is therefore the biggest strike we can insert with a call price equal to
        intrinsic value such that call prices remain convex. Then set $K^{\min} := 1/10 \min_j K^*_j$.
        
        To determine $K^{\max}$ find per expiry the strike $K^\#_j$ where the line from $(K^{N_j-2},C^{N_j-2}_j)$ through $(K^{N_j-1},C^{N_j-1}_j)$ intersects with zero, i.e.~$C^{N_j-1}_j + dC^{N_j-2}_j ( K^\#_j - K^{N_j-1}_j) = 0$.
        Then define $K^{\max}:=1.5 \max_j K^\#_j$.
} when working with real market data to satisfy the conditions on the strikes
in the previous theorem is to first consider strikes $0\ll K^2_j<\cdots<1<\cdots<K_j^{N_j-1}$ with calls $C^i_j > (1-K_j^i)^+$ and $1>C^2_j>\cdots > C^{N_j-1}_j>0$
which satisfy $-1 < dC^2_j \leq \cdots \leq dC_j^{N_j-1} < 0$,
and then to insert the missing data as follows:
\begin{enumerate}
    \item $K^0_j:=0$ with $C^0_j:=1$;
    \item $K^1_j := K^{\min} \ll \min_j K^2_j$ with $C^1_j:=1-K^{\min}$ such that $-1=dC^0_j < dC^1_j < dC^2_j$ for all $j$;
    \item $K^N_j :- K^{\max} \gg \max_j K^{N_j-1}_j$ with $C_j^{N_j}=0$ such that $dC^{N_j-2}_j < dC_j^{N_j-1} < 0$ for all $j$.
\end{enumerate}
\end{remark}

Given a discrete state martingale~$X$ as per above, let us denote with slight abuse
of notation by $X_j:=X_{T_j}$
the corresponding discrete state and discrete time martingale Markov chain.
Given strikes $K_j=(K^1_j,\ldots,K^{N_j}_j)'$ and a Markov chain $(X_j)_{j=0,\ldots,M}$ with
$\P(X_j=K^i_j)=:p_j^i$, define the transition probabilities
\begin{equation}\label{eq:transX}
P_{j|j-1}^{i|\ell}:=\P\!\left[X_j=K^i_j \mid X_{j-1}=K^\ell_{j-1}\right], \qquad i=1,\ldots,N_j;\ \ell=1,\ldots,N_{j-1}.
\end{equation}
Let $P_{j|j-1}\in\R^{N_j\times N_{j-1}}$ denote the matrix with entries $(P_{j|j-1})_{i,\ell}=P_{j|j-1}^{i|\ell}$.
(Thus columns correspond to the conditioning state $\ell$.)
\begin{definition}[Martingale transition operator]\label{def:transiop}
A matrix $P\in\R^{N_j\times N_{j-1}}$ is called a \emph{(discrete) martingale transition operator}
with respect to strikes $K_j$ and $K_{j-1}$ if
\begin{enumerate}
    \item it is column-stochastic: $P\ge 0$ and $1' \cdot P=1'$;
    \item it satisfies the martingale property: $K_j'\cdot P=K_{j-1}'$,
\end{enumerate}
where we used ``$\cdot$'' to denote
matrix/vector multiplications. Equivalently, for each conditioning index $\ell\in\{1,\ldots,N\}$, the above properties translate into 
\[
    \sum_{i=1}^{N_j} P^{i|\ell}=1,
    \qquad\text{and}\qquad
    \sum_{i=1}^{N_j} K^{i}_j\,P^{i|\ell}=K^{\ell}_{j-1}.
\]
\end{definition}
\begin{definition}[Martingale Density]\label{def:martdens}
    A~sequence $p=(p_j)_{j=0,\ldots,M}$ of densities $p_j\in \R^{N}_{\geq 0}$ each defined for expiries $0=T_0<\cdots<T_M$ over strikes $K_j$ with unit mean is called a \emph{martingale density}
    if there exist\footnote{
         As pointed out in~\cite{buehler2006expensive} discrete time transition operators
    are by no means unique for given marginal distributions.    
    } martingale transition operators $P_{j|j-1}\in\R^{N\times N}$  such that
    $$
        p_j = P_{j|j-1} p_{j-1} 
    $$
    written out as
    $$
        p_j^{i_j} = \sum_{i_{j-1}=1}^{N_{j-1}} P_{j|j-1}^{i_j|i_{j-1}} p_{j-1}^{i_{j-1}} \ .
    $$     
    (The initial density $p_0=(1)$ is trivially defined for $N_0=1$ and strike $K_0^1:=1$).
\end{definition}

\begin{remark}
    The martingale density~$p$ for an arbitrage-free market of call prices $C$ is given as
    \begin{equation}\label{def:p}
        p_j^i := dC_j^i - dC_j^{i-1} \geq 0 \ .
    \end{equation}
    for $i=1,\ldots,N_j$
    with $dC^i_j$ defined in~\eqref{def:dC}. We also let $p_0^1:=1$. (Theorem~\ref{th:dlvtrs} on page~\pageref{th:dlvtrs} shows how to construct respective transition operators.   
\end{remark}

\subsection{Linear Arbitrage-Free Call Price Surfaces}\label{sec:lin_surface}

We now recall the well-known result that linear interpolation in strike and expiry yields an arbitrage-free -- but not smooth -- interpolator, to motivate our later results.

\begin{definition}[Linear Interpolation]\label{th:linintp}
The linear interpolator in strikes $K$ for expiry $T_j$ is
$$
    \bar C_j(K) := \sum_{i=0}^{N_j-1} \Big\{  dC^i_j ( \underbrace{ K-K_j^{i+1}) }_{\leq 0} + C^{i+1}_j \Big\}_{K_j^i\leq K<K_j^{i+1}} \ .
$$
with the trivial $\bar C_0(K):=(1-K)^+$. To interpolate in time without causing
calendar arbitrage,
any increasing interpolation along fixed strikes will be sufficient. Assume 
therefore that ~$\alpha_j:(T_{j-1},T_j]\rightarrow [0,1]$ 
is increasing with
with $\alpha_j(T_{j-1})=0$ and $\alpha_j(T_j)=1$. Then let
\begin{equation}\label{eq:intT}
    \bar C(T,K) := \bar C_j(K) \alpha_j(T) + (1-\alpha_j(T) ) \bar C_{j-1}(K) \ .
\end{equation}
\end{definition}   

 \begin{remark}
    A~naive choice for the time-interpolation functions~$\alpha_j$ is
    linear interpolation
    \begin{equation}\label{eq:linintp}
         \alpha_j(T):= \frac{ T\wedge T_j - T_{j-1} }{ T_j - T_{j-1} } \ .
    \end{equation}
    A common alternative is to interpolate smoothly in implied Black-Scholes ATM variance.
    To this end let $W_j$ be the Black-Scholes implied variance for $C_j(1)$. 
    Define an increasing interpolation function $W(T)$ with $W(T_j)=W_j$
    for example an interpolating monotone $C^1$ spline\footnote{Monotone splines are implemented in 
    scipy's \texttt{PchipInterpolator}.} Then set
    $$
    \alpha_j(T) := \frac{ \mathrm{Call}\big(1, 1, W(T)\big) - C_{j-1}(1) }
    { C_j(1) - C_{j-1}(1) } \1_{T\in(T_{j-1},T_j]} \ .
    $$
    
\end{remark}

\noindent
For simplicity we will continue refer to $\bar C$ as ``linear interpolation'' even if~$\alpha$ may not be linear.

\begin{proposition}\label{prop:linintp}
    It holds that 
    \begin{equation}
        \bar C_j(K) =
             \sum_{i=1}^{N} p_j^i\, (K^i_j-K)^+ \label{eq:hatCputs} \ .
    \end{equation}
    with the martingale density~$p$ as
    defined in~\eqref{def:p}.
\end{proposition}
\begin{proof}
We notice that
$$
    C^{i+1}_j = 
    \left( C^{i+1}_j - C^{i+2}_j \right) + \left( C^{i+2}_j - C^{i+3}_j \right) + \cdots + \left( dC_j^{N-1} - dC_j^{N}\right)
    =   
    \sum_{\ell=i+1}^{N-1} dC^\ell_j ( \underbrace{K_j^\ell - K_j^{\ell+1}}_{\leq 0} )
$$
hence
\begin{eqnarray}
    \bar C_j(K) & = & \nonumber
    \sum_{i=0}^{N_j-1} \1_{K_j^i\leq K<K_j^{i+1}} \left\{
    dC^i_j (  \underbrace{ K-K_j^{i+1} }_{\leq 0})  +
        \sum_{\ell=i+1}^{N_j-1} dC^\ell_j ( \underbrace{K_j^\ell - K_j^{\ell+1}}_{\leq 0} ) \right\}\\
    & = &
        \sum_{i=0}^{N_j-1} dC_j^i \Big( (K_j^i-K)^+ - (K^{i+1}_j-K)^+  \Big) \label{eq:hatCsumdC}\\
    & = &
        \underbrace{ dC^0_j (K_j^0 - K)^+ }_{=0}
        -
        dC^{N_j-1}_j (K_j^{N_j} - K)^+
        +
        \sum_{i=1}^{N_j-1} \Big( \underbrace{ dC_j^i - dC_j^{i-1} }_{\geq 0} \Big) (K_j^i-K)^+ \nonumber
        \\
    & = & 
        \sum_{i=1}^{N_j} \Big( \underbrace{ dC_j^i - dC_j^{i-1} }_{= p_j^i \geq 0} \Big) (K_j^i-K)^+  \ = \eqref{eq:hatCputs}  \  .
\end{eqnarray}
where we used that we had defined~$dC_j^{N_j}:=0$.

The fact that $p_j$ is non-negative and sums up to~1 is clear. To show that $p_j$ has unit expectation follows from $\bar C_j(0)=1$.
\end{proof}
\begin{proposition}
    Assume the market is arbitrage-free in the sense of theorem~\ref{th:noarb}.
    
    Then, linear interpolation
    $\bar C$
    is arbitrage-free.
\end{proposition}
\begin{proof}
\hspace{1cm}

\begin{enumerate}
\item 
Equation~\eqref{eq:hatCsumdC} shows
$
    \bar C_j(0) = \sum_{i=0}^{N-1} dC_j^i ( K^i_j - K^{i+1}_j ) = 0 - (-1) = 1 
$.
\item
With $\partial K (K^i_j-K)^+|_{K=0} = \1_{K>K_j^0}$ we also see
$
    \partial_K \bar C_j(K)|_{K=0} = dC^0_j = -1 
$.
\item 
$\lim_{K\uparrow \infty} \bar C_j(K)=0$ is satisfied by constriction since $C_j(K^{N_j}_j)=0$.
\item Equation~\eqref{eq:hatCputs} shows that $\bar C_j$ is convex.
\item Considering~\eqref{eq:arbfreeinT} and~\eqref{eq:intT} $\bar C$ is increasing in $T$. Note that a convex
combination of two convex functions remains convex.
\end{enumerate}
\end{proof}

\subsubsection*{Why not Stay Linear}

Linear interpolation is very natural, and has the desirable property to exist if and only if the market (on a suitable strike grid)
is arbitrage-free. However, as pointed out in~\cite{buehler2006expensive}, linear interpolation is ``the most expensive''
interpolation of option prices which is consistent with the anchor prices.

In practice that means implied volatilities between interpolated strikes is too high. Figure~\ref{fig:linexp}
illustrates this effect with an example on real data.
\begin{figure}[H]
    \centering
    \includegraphics[width=0.4\linewidth]{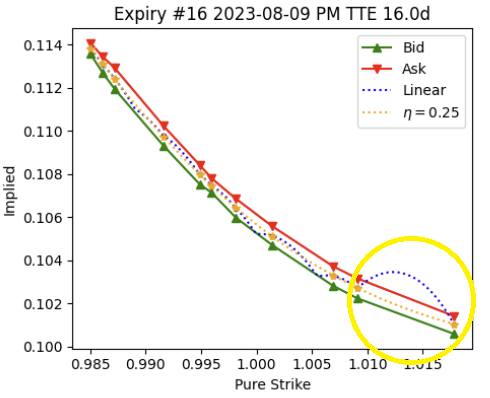}
    \caption{Effect on implied volatility of linearly interpolating call prices. The ``$\eta=0.25$'' interpolation
    corresponds to SANOS with smoothness~$0.25$ and visibly provides a superior fit.}
    \label{fig:linexp}
\end{figure}

\section{Smooth Arbitrage-Free Call Price Surfaces}\label{sec:smootharbfree}

In this section we present the main results of this article: 
While linear interpolation is a simple and powerful tool construct arbitrage-free call price surfaces, it essentially represents
a state-discrete underlying process. Indeed, $\bar C$ is not smooth with $\partial_{KK} \bar C(T,K)=0$ for all $K\not\in \{K^1_1,\ldots,K^{N_M}_M\}$ or $T\not\in\{ T_1,\ldots,T_M \}$.
The density $p(T,K)$ for the full surface can only have mass in $(T,K^i_j)$.

In most practical applications we aim to have a smooth representation of call price surfaces which allows pricing any option at any strike and expiry. Linear interpolation
might be mathematically valid, but it is also extreme: indeed, linear interpolation of a set of arbitrage-free market prices provides ``the most expensive interpolation scheme
available''~\cite{buehler2006expensive}. That means it is not helpful for inferring insights on values of options between quoted strikes and expiries. From a practical perspective a relatively sparse linear interpolator does not lend itself well to the use of continuous state methods, chiefly Dupire's famous local volatility~\cite{dupire1994pricing},
given on a smooth grid as
$$
    \sigma(T,K)^2 := 2 \frac{ \partial_T C(T,K) }{ K^2 \partial_{KK} C(T,K) dT } \ .
$$
(In~\cite{DLV} we presented with ``discrete local volatilities'' an alternative
to Dupire's local volatility which works with large time steps and discrete market data.)
 
We now present an alternative to interpolate a set of arbitrage-free market call prices with smooth functions such that the resulting surface is itself also smooth.
We start with the following reference model which performs surprisingly well
despite its simplicity. We managed to fit it to a wide range of SPX market data.
Section~\ref{sec:implementation} on page~\pageref{sec:implementation} 
discusses implementation.

\begin{theorem}[Smooth Call Prices]\label{th:simplersmooth}
    Let $Y=(Y_t)_{t\geq 0}$ be an atomless positive martingale with full support over
    $(0,\infty)$ and unit mean.
    Set $Y_j:=Y_{T_j}$.
    Let $q=(q_0,\ldots,q_M)$ be an at martingale density as defined in definition~\ref{def:martdens} 
    and let
    \begin{equation}\label{eq:CasSumC}
        \hat C_j(K) := \sum_{i=1}^{N} q_j^i\ \E[ ( K^i_j Y_j - K )^+ ] \ .        
    \end{equation}
    Using a time-interpolation scheme~$\alpha$ such as~\eqref{eq:intT} we define call prices for expiries $T\in(T_{j-1},T_j]$
    as
    \begin{equation}\label{eq:alpha_prod}
        \hat C(T,K) := \alpha_j(T) \sum_{i=1}^{N_j} q_j^i\ \E[ ( K^i_j Y_T - K )^+ ] + (1-\alpha_j(T)) \sum_{i=1}^{N_{j-1}} q_{j-1}^i\ \E[ ( K^i_{j-1} Y_T - K )^+ ] \ .
    \end{equation}
    For $T>T_M$ we extrapolate using $Y$:
    \begin{equation}\label{eq:extra}
        \hat C(T,K) := \sum_{i=1}^{N_M} q_j^i\ \E[ ( K^i_M Y_T - K )^+ ]  \ .
    \end{equation}

    Then $\hat C$ is a smooth call price function.
    
    Moreover, if~$q=p$, then $\hat C_j \geq \bar C_j$ with equality if $Y\equiv 1$ in which
    case~$\hat C$ 
        reduces to linear interpolation.
\end{theorem}

\noindent
The discrete-time martingale representation for~$\hat C_0,\ldots,\hat C_M$ is given by
\begin{equation}\label{eq:Zmart}
    Z_j := X_j Y_j     
\end{equation}
in the sense that $\E[ (Z_j-K)^+] = \hat C_j(K)$
where $X\sim q$ is the discrete state Markov martingale~\eqref{eq:transX}.

\begin{proof}[Proof of the theorem]
Strictly speaking~\eqref{eq:Zmart} is sufficient to prove that 
$\hat C$ is arbitrage-free. For clarity with still provide a brief proof:
    \begin{enumerate}
        \item $\hat C_j(0) = \sum_i q_j^i\, K^i_j \E[ Y_j ] = 1$ because $q_j$ and $Y_j$ has unit mean.
        \item $\partial_K \hat C_j(0)= - \sum_i q_j^i\, \P[ K^i_j Y_j > 0 ] = -1$ because $Y_j>0$ and~$q_j$ has unit mean.
        \item $\lim_{K\uparrow \infty} \hat C_j(K) = 0$ because $\E[Y_j]=1$ and dominated convergence applies.
        \item $\partial_{KK} \hat C_j(K) > 0$ by construction for $K>0$.
        \item To show that call prices are increasing in time, consider
        \begin{eqnarray}
             \sum_{i=1}^{N_{j+1}} q_{j+1}^i\, \E[ ( K^i_{j+1} Y_{j+1}- K )^+ ] \nonumber
             & \supstack{(*)}\geq &
             \sum_{i=1}^{N_{j+1}} q_{j+1}^i\, \E[ ( K^i_{j+1} Y_j - K )^+ ] \nonumber\\
              & \supstack{d\P^Y := Y d\P}= &
            \E^Y\left[  \sum_{i=1}^{N_{j+1}} q_{j+1}^i\, ( K^i_{j+1}  - K/Y_j  )^+ \right] \nonumber\\
             & \supstack{(**)}\geq&
            \E^Y\left[  \sum_{i=1}^{N_j} q_j^i\,  ( K^i_j - K/Y_j  )^+ \right] \nonumber\\
             & =&
             \sum_{i=1}^{N_j} q_j^i\, \E[ ( K^i_j Y_j - K )^+ ] \nonumber
        \end{eqnarray}
        where $(*)$ follows because~$Y$ is a martingale, and $(**)$ follows as before because
        $q$ is a martingale density which means, in particular,
        $$
        \sum_{i=1}^{N_{j+1}} q_{j+1}^i\, ( K^i_{j+1} - K )^+ 
        \geq
             \sum_{i=1}^{N_j} q_j^i\, ( K^i_j - K )^+ 
        $$
        for all $K$. Applying the linear expectation operator~$\E^Y[\cdot]$ preserves the inequality.

        It is clear that monotonicity for arbitrary $T$ is preserved by our time-interpolation  scheme.
    \end{enumerate}
\end{proof}
\begin{remark}
Our approach is a special case of the following concept: let $X=(X_t)_t$ and $Y=(Y_t)_t$ be two independent martingales.
Then, $Z_t:=X_t Y_t$ is again a martingale and has call prices
$$
    \E[ (Z_t-K)^+ ]  = \E[ ( X_tY_t - K)^+ ] = \int_0^\infty \E[ ( x\,Y_t - K)^+ ]\, \P[ X_t=x ]\,dx \ ,
$$
mirroring~\eqref{eq:Zmart}. It is worth noting that any classic arbitrage-free pricing model~$Y$ can be used for our model. Heston's model~\cite{heston1993closed} for example provides a rich surface with pronounced skew which can be used
as a background martingale.
\end{remark}

\subsection{Production Model}\label{sec:prod_model}

For our proposed production model whose implementation we will discuss in section~\ref{sec:implementation}
we propose a very simplistic structure which still fits market data remarkably well: we will basically assume $Y$ is log-normal.
Assume that $V_1<\cdots<V_M$ are ATM variances from the market for model expiries $0<T_1<\cdots<T_M$ and let $\eta\in[0,1)$ be a parameter which controls the smoothness
of our model. Given model strikes $K^i_j$ our production model is then given as
\begin{equation}\label{eq:prod_model}
        \hat C_j(K) := \sum_{i=1}^{N_j} q_j^i\, \mathrm{Call}\big( K^i_j,\, K,\, \eta V_j \big)
\end{equation}
We then solve for~$q$ to globally fit the model within market bid/ask prices for market expiries and strikes.

The parameter~$\eta$ allows controlling the desired smoothness of~\eqref{eq:prod_model}: if it is set to zero, 
the model is equivalent to the linear model discussed in section~\ref{sec:lin_surface}.
By increasing~$\eta$ we can then smoothen our model. In the extreme case $\eta=1$ the model can only fit ATM options.
Figures~\ref{fig:example1DTE},~\ref{fig:example35DTE} and~\ref{fig:example1249DTE} illustrate three
different cases of fits with $\eta=0$, $\eta=0.25^2\approx 0.06$ and~$\eta=0.5^2=0.25$. The latter is our recommended
default setting.

\begin{figure}[H]
    \centering
    \includegraphics[width=0.9\linewidth]{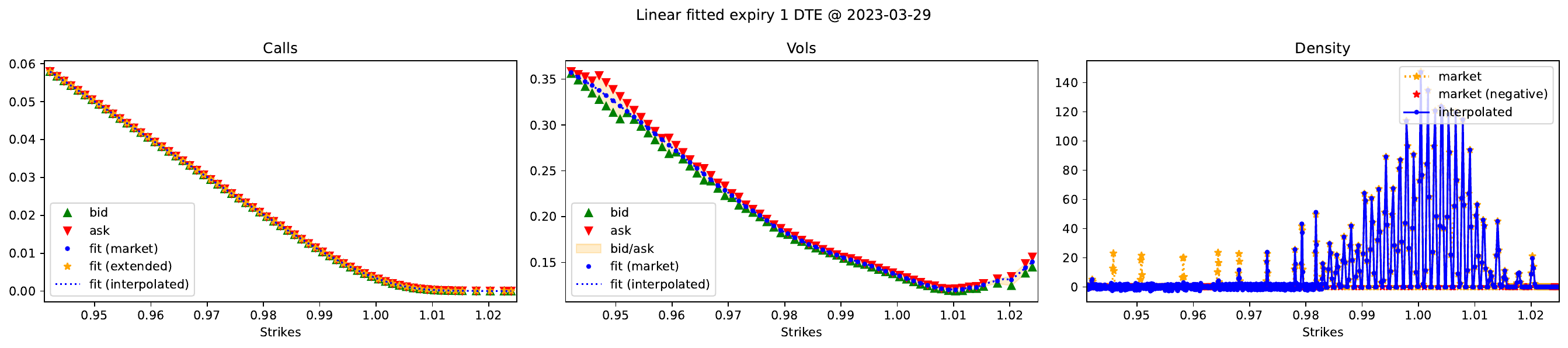}\\
    \includegraphics[width=0.9\linewidth]{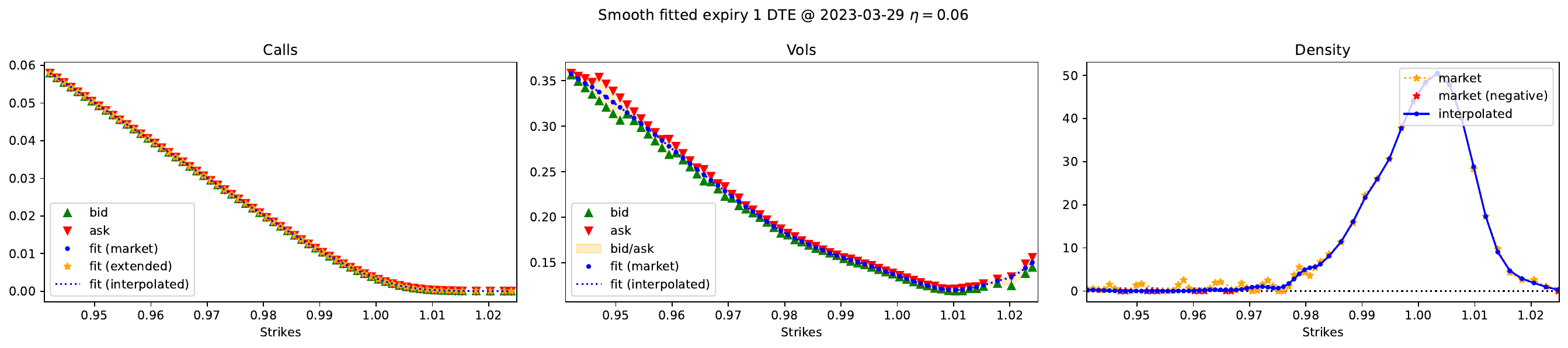}\\
    \includegraphics[width=0.9\linewidth]{figs/smoothvol_2023-03-29_dte-1DTE_eta-0.25.pdf}
    \caption{Fit of our model with different smoothness parameters~$\eta$ to~1~day to expiry options on SPX. The right hand side shows the density which is much smoother
    for the case~$\eta=0.25$ while actually retaining a decent fit: the biggest fitting error is~40\% of vol spread.}
    \label{fig:example1DTE}
\end{figure}

\begin{figure}[H]
    \centering
    \includegraphics[width=0.9\linewidth]{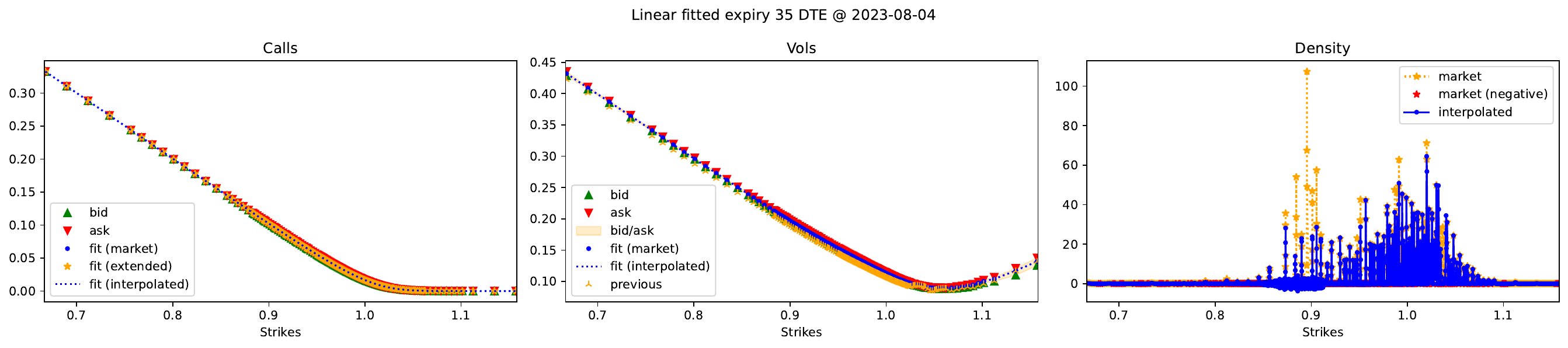}\\
    \includegraphics[width=0.9\linewidth]{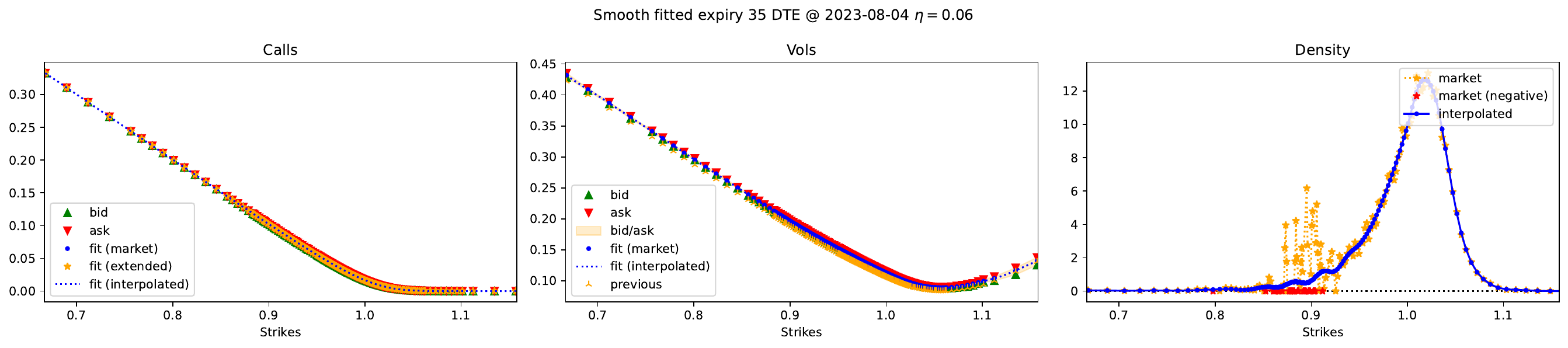}\\
    \includegraphics[width=0.9\linewidth]{figs/smoothvol_2023-08-04_dte-35DTE_eta-0.25.pdf}
    \caption{Fit of our model with different smoothness parameters~$\eta$ to SPX options with~35 days to expiry. In this case~$\eta=0.06$ fits the market within bid/ask
    while imposing sufficient smoothness on the model.
    }
    \label{fig:example35DTE}
\end{figure}

\begin{figure}[H]
    \centering
    \includegraphics[width=0.9\linewidth]{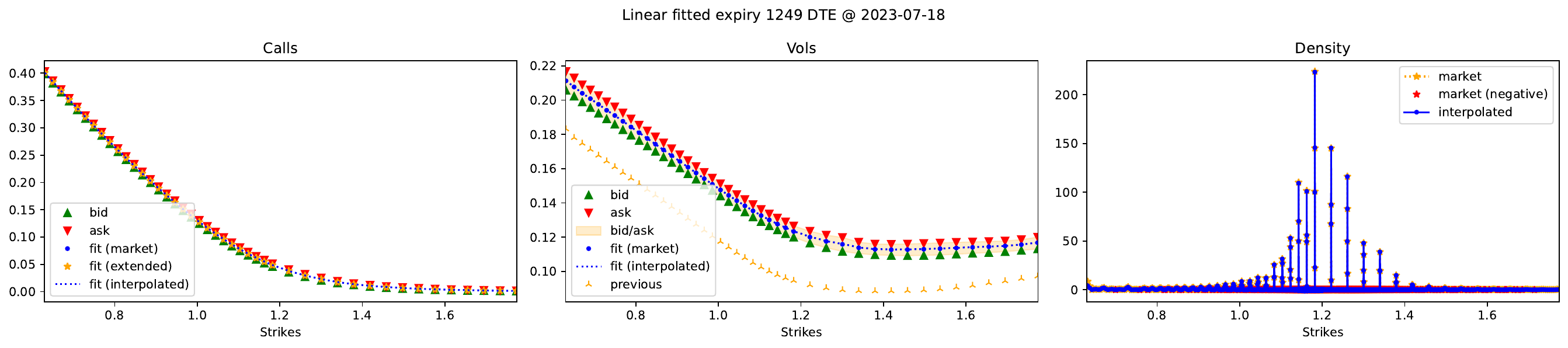}\\
    \includegraphics[width=0.9\linewidth]{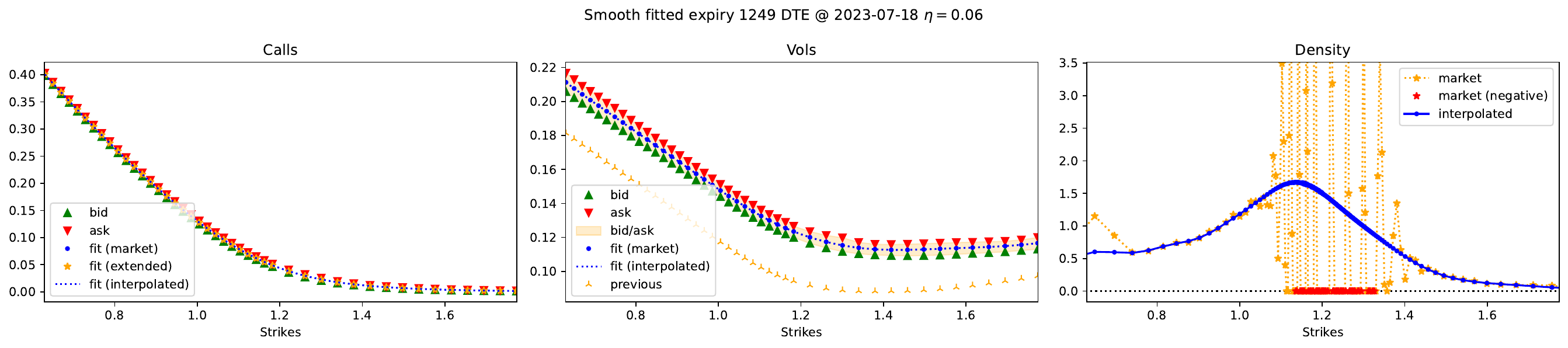}\\
    \includegraphics[width=0.9\linewidth]{figs/smoothvol_2023-07-18_dte-1249DTE_eta-0.25.pdf}
    \caption{Fit of our model with different smoothness parameters~$\eta$ to SPX options with~1249 days to expiry. 
    }
    \label{fig:example1249DTE}
\end{figure}

It should be noted that extrapolation within the range where implied volatilities can numerically
be computed extrapolation is also fairly natural. Figure~\ref{fig:extrapolatio}
illustrates this with an example.

\begin{figure}[H]
    \centering
    \includegraphics[width=0.9\linewidth]{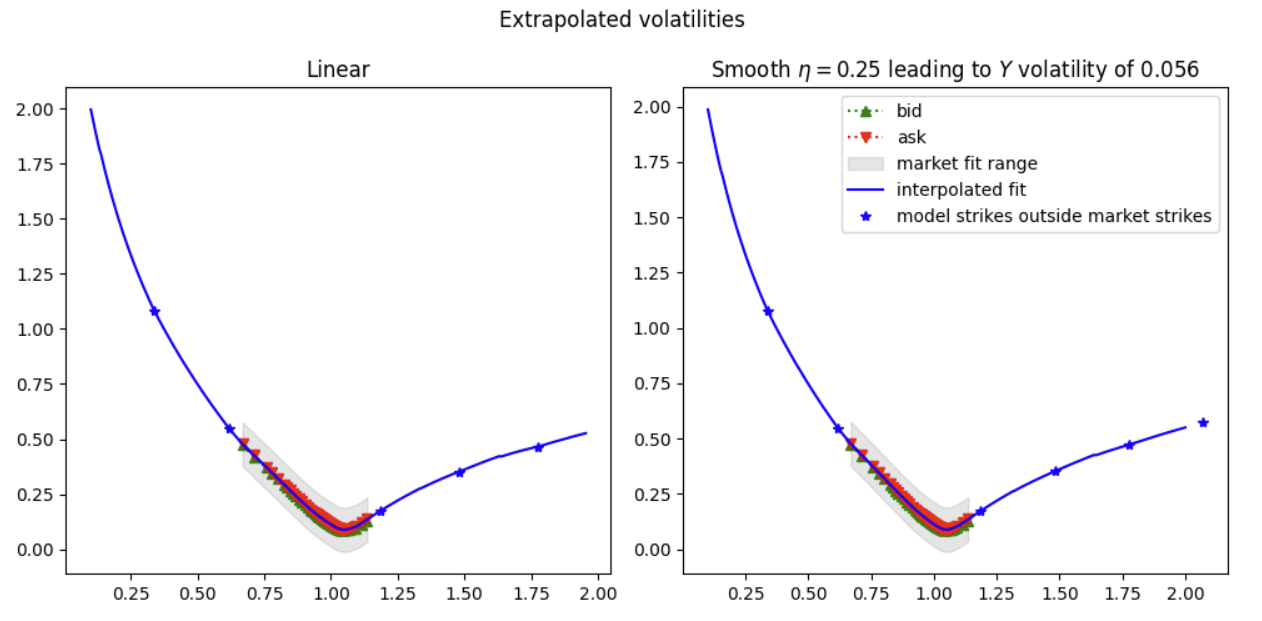}
    \caption{An example for extrapolation. The blue marks indicate
    added model strikes outside the observed market strike range.
    }
    \label{fig:extrapolatio}
\end{figure}

\begin{remark}[Generalized Interpolation in Strike]
    The proof of theorem~\ref{th:simplersmooth} that the function~$K\mapsto \hat C_j(K)$ is
    arbitrage-free in sense of satisfying~\ref{it:caf_C0_1}-\ref{it:caf_cvx}
    of definition~\ref{th:contarbfree} remains valid if we replace the common~$Y$ with a martingale by 
    strike:
    $$
        {\hat {\hat C}}_j(K) := \sum_{i=1}^{N} q_j^i\ \E[ ( K^i_j Y^i_j - K )^+ ] \ .
    $$
    In the log-normal case translates to using different variances
    per strike:
    \begin{equation}\label{eq:inhC}
        {\hat {\hat C}}_j(K)= \sum_{i=1}^{N} q_j^i\, \mathrm{Call}\big( K^i, K, \eta V_j^i \big) \  .
    \end{equation}
    This has also been found in~\cite{elalaoui2023}.
    However, in this case we have not found a numerically efficient 
    condition on the $V$'s to ensure absence of arbitrage in time for all strikes
    except in the case where the anchor call prices themselves are generated
    by a common martingale~$Y$.
    For example it is not sufficient that the two expiries
    are in order at the market strikes
    as the synthetic example in figure~\ref{fig:vioimh} shows.

\begin{figure}[H]
    \centering
    \includegraphics[width=0.5\linewidth]{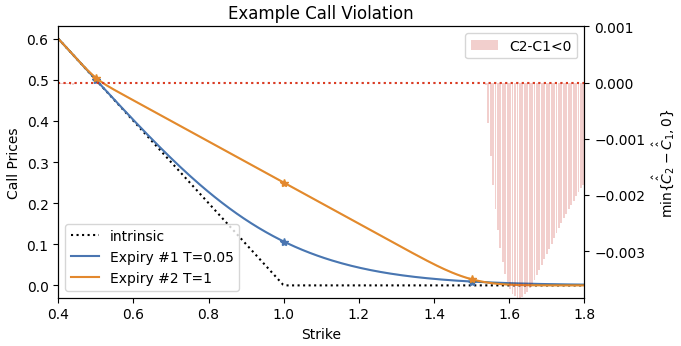}
    \caption{
        Example of the idea~\eqref{eq:inhC} where the call prices at two expiries
        are increasing in the market strikes~$K=(0.5,1,1.5)$ but not
        at the extrapolated strikes $K\geq 1.6$ (the error is plotted at the right hand axis). The example uses constant
        log-normal volatilities $(0.05,1.2,0.05)$ for the three strikes for
        both expiries. The densities are $q_1=(0,1,0)$ and $q_2=(0.5,0,0.5)$.
    }
    \label{fig:vioimh}
\end{figure}

    In section~\ref{sec:genr} we present a different, formally correct
    generalization to different~$Y^i_j$'s per strike and expiry; however
    it is numerically much less efficient than our production model presented here
    and prohibitively expensive to run.
\end{remark}

\subsection{Efficient Parametrization}

Fitting the  model~\eqref{eq:CasSumC} given variances~$0<V_1<\cdots<V_M$
means finding a martingale density with marginals~$q_j$ such that the model
prices are within bid/ask spread of observed market data. This can easily be
implemented using linear programming as is shown in section~\ref{sec:implementation}
which also covers the case of different strikes per expiry (which is the
case in practise).

However, if we aim to use our parametrization in generative models then
its representation in terms of a  martingale density is not convenient as 
the marginals~$q_1,\ldots,q_M$
have to satisfy the linear constraints of definition~\ref{def:martdens}. 
For such applications we propose parameterizing the martingale density in terms of
the \emph{discrete local volatilities} introduced in~\cite{DLV}. This parametrization
has been used to great effect in  multi-asset arbitrage-free 
option surface simulators in~\cite{wiese2021multiassetspotoptionmarket}, albeit
on discrete grids of strikes.
With our approach here
such generative model is now able to generate smooth arbitrage-free option surfaces.

We start first with the case of homogeneous strikes $0=K^0<K^1<\cdots<1<\cdots<K^N$
across expiries which is much simpler and an ingredient into the inhomogeneous case.
This material is from~\cite{DLV} with notations aligned.

\begin{theorem}[Transition Operators from Discrete Local Volatilities]\label{th:dlvtrs}
Let~$\Sigma^i_j\geq 0$ for $i=2,\ldots,N-1$ and $j=1,\ldots,M$ be
a surface of input \emph{discrete local volatilities}.
Let
$$
    \gamma^{i+}_j := \frac1{\frac12(K^{i+1}-K^{i-1})} \frac1{ K^{i+1} - K^i } \ \ \ \mbox{and}  \ \ \ \gamma^{i-}_j := \frac1{\frac12(K^{i+1}-K^{i-1})} \frac1{ K^i - K^{i-1} } 
$$
and define
\begin{equation}\label{eq:wDLV}
    w^{i\pm}_j := \frac12 (\Sigma^i_j K^i)^2 (T_j-T_{j-1})\ \gamma^{i\pm}_j 
\end{equation}
and $w^{1\pm}_j = w^{N\pm}_j : = 0$.
Then the inverse~$Q_{j|j-1}$ of the tri-band matrix
\begin{equation}\label{eq:implicit}
    Q^{-1}_{j|j-1} :=
    \left( 
        \begin{NiceArray}{c:c:c:c:c}
             1 & -w_j^{2-} \\
             \hdottedline \\
             0 & 1+w^{2-}_j+w^{2+}_j \\
             \hdottedline \\
              & -w^{2+}_j & \ddots & -w^{(N-1)-}_j \\
             \hdottedline \\
              & & & 1+w^{(N-1)-}_j+w^{(N-1)+}_j & 0 \\
             \hdottedline \\
             & & & -w^{(N-1)+}_j & 1 
        \end{NiceArray}
    \right)
    \in \R^{N\times N}
\end{equation}
is a martingale transition operator.

Starting in~$q_0=(1)$ we then obtain the corresponding martingale densities
$$
    q_j :=  Q_{j|j-1} \cdot q_{j-1} \ . 
$$
\end{theorem}

\noindent
The proof is provided in appendix~\ref{sec:appendix_DLV} on page~\pageref{sec:appendix_DLV}.

\begin{remark} The tri-band matrix $Q^{-1}_{j|j-1}$ can be written in terms of a vector
$\Sigma_j=(0, \Sigma^2_j,\ldots, \Sigma^{N-1}_j, 0)$
as follows:
$$
    Q^{-1}_{j|j-1} = E + \Omega\ \Sigma_j{}^2 (T_j-T_{j-1})
$$
(where the right hand multiplication of the matrix with the vector is column-wise) 
for pre-computed matrices
$$
    \Omega  := 
    \left( 
        \begin{NiceArray}{c:c:c:c:c}
             0 & -\omega_j^{2-} \\
             \hdottedline \\
             0 & \omega^{2-}_j+\omega^{2+}_j \\
             \hdottedline \\
              & -\omega^{2+}_j & \ddots & -\omega^{(N-1)-}_j \\
             \hdottedline \\
              & & & \omega^{(N-1)-}_j+\omega^{(N-1)+}_j & 0 \\
             \hdottedline \\
             & & & -\omega^{(N-1)+}_j & 0
        \end{NiceArray}
    \right)
    \in \R^{N\times N}
$$
with
$$
    \omega^{i+}_j := \frac{
                        (K^i)^2 
                        }{(K^{i+1}-K^{i-1})( K^{i+1} - K^i) } 
                        \ \ \ \mbox{and}  \ \ \ 
    \omega^{i-}_j := \frac{
                        (K^i)^2 
                        }{(K^{i+1}-K^{i-1})( K^i - K^{i-1} )}  \ .
$$
\end{remark}

We now expand our definition to inhomogeneous strikes. 
Assume therefore we are given general strikes $0=K_0<K^{\min}=K^1_j<\cdots<1<\cdots<K^{N_j}_j=K^{\max}$
    and discrete local volatilities $\Sigma_j=(\Sigma^2_j,\ldots,\Sigma^{N_j-1}_j)'$
    which each imply a transition operator $Q_{j|j-1}\in\R^{N_j\times N_j}$, with the caveat that the latter is defined
    to transition a density $\bar q_{j-1}$ defined over strikes $K_j$ into a density $q_j$, also
    defined over $K_j$. However, the previous density $q_{j-1}$ is given over strikes $K_{j-1}$.
We apply the idea from~\cite{DLV} to apply linear interpolation between expiries i.e.
\begin{equation}\label{eq:barcintlp}
    \bar C_{j-1}^\ell := \sum_{i=1}^{N_{j-1}} q_{j-1}^i ( K^i_{j-1} - K^\ell_j )^+    
\end{equation}
for $\ell=0,\ldots,N_j$. We note that $\bar C^{0}_{j-1} = 1$ and $\bar C^{N_j}_{j-1} = 0$. These call prices are arbitrage-free and imply a density
\begin{eqnarray*}
    \bar q^\ell_j & := &
        \frac{ \bar C_{j-1}^{\ell+1} - \bar C_{j-1}^{\ell} }{ K^{\ell+1}_j - K^{\ell}_j }
        - 
        \frac{ \bar C_{j-1}^{\ell} - \bar C_{j-1}^{\ell-1} }{ K^{\ell}_j - K^{\ell-1}_j }
        \\
        &  = &
        \sum_{i=1}^{N_{j-1}} q_{j-1}^i
        \left( 
        \frac{ ( K^i_{j-1} - K^{\ell+1}_j )^+ - ( K^i_{j-1} - K^\ell_j )^+ }{ K^{\ell+1}_j - K^{\ell}_j }
        -
        \frac{ ( K^i_{j-1} - K^{\ell}_j )^+ - ( K^i_{j-1} - K^{\ell-1}_j )^+ }{ K^{\ell}_j - K^{\ell-1}_j }
        \right)
\end{eqnarray*}
for $\ell=1,\ldots,N_j$ where $K^{N_j+1} > K^{N_j}$ is arbitrary. The term on the right is $L^{\ell,i}_{j|j-1}$
for $\ell=1,\ldots,N_j$ and $i=1,\ldots,N_{j-1}$. That means:

\begin{theorem}\label{th:lintransition}
    Assume that $q_{j-1}$ is a density over $K_{j-1}$ with unit mean.
    Define the matrix $L_{j|j-1} \in \R^{N_j\times N_{j-1}}_{\geq0}$ with elements
    \begin{equation}\label{eq:L}
        L_{j|j-1}^{\ell,i} := 
        \frac{ ( K^i_{j-1} - K^{\ell+1}_j )^+ - ( K^i_{j-1} - K^\ell_j )^+ }{ K^{\ell+1}_j - K^{\ell}_j }
        -
        \frac{ ( K^i_{j-1} - K^{\ell}_j )^+ - ( K^i_{j-1} - K^{\ell-1}_j )^+ }{ K^{\ell}_j - K^{\ell-1}_j }
        \geq 0
    \end{equation}
    for $\ell=1,\ldots,N_j$ and $i=1,\ldots,N_{j-1}$.
    Then
    $$
        \bar q_{j-1} = L_{j|j-1} q_{j-1} \in \R^{N_j} 
    $$
    is a density with unit mean over~$K_j$ where $\bar q_{j-1}$ represents the density
    implied by linearly interpolated call prices computed using $p_{j-1}$ at strikes~$K_j$.
    
    Finally,
    $$
        K_j L_{j|j-1} = K_{j-1} \ .
    $$
\end{theorem}
\noindent 
The proof that $1L_{j|j-1}=1$ and $ K_j L_{j|j-1} = K_{j-1}$ is in the appendix.

\begin{theorem}[Implied Discrete Local Volatilities]\label{th:dlvtoQ}
Assume that the discrete call prices $C^i_j$ for $j=1,\ldots,M$ and $i=1,\ldots,N_j$
are arbitrage-free and define
    \begin{equation}\label{eq:dlv2}
        \Sigma^i_j := \sqrt{ 2 \frac{  \Theta_j^i }{ K^i_j{}^2 \, \Gamma_j^i } }
        \ \ \ \mbox{with}\ \ \
        \Theta_j^i :=\frac{ C_j^i - \bar C_{j-1}^i }{ T_j - T_{j-1}}
        \ \ \ \mbox{and}\ \ \ 
        \Gamma_j^i := \frac{ dC_j^i - dC_j^{i-1} }{ \frac12 (K_j^{i+1} - K_j^{i-1}) } 
    \end{equation}
    for $j=1,\ldots,M$ and $i=2,\ldots,N_j-1$ with $0/0:=0$. Here, $\bar C_{j-1}^i$
    is again given by linear interpolation of~$C_{j-1}$ as in~\eqref{eq:barcintlp}.
    
    In this case
    \begin{equation}\label{eq:wself}
        w^{i\pm}_j =  \frac{ C_j^i - \bar C_{j-1}^i }{ p_j^i }  \bar \gamma^{i\pm}_j
        \ \ \ \mbox{with} \ \ \
        \bar \gamma^{i+}_j:=\frac1{K_j^{i+1}-K_j^i}
        \ , \ \ \
        \bar \gamma^{i-}_j:=\frac1{K_j^i-K_j^{i-1}} \ .
    \end{equation}
Then, the transition operator~$Q_{j|j-1}$ reproduces discrete input call prices
$C^1_1,\ldots,C^N_M$ in the sense
that $p_j = Q_{j|j-1} \tilde p_{j-1}$.
\end{theorem}
\noindent
The proof is also provided in appendix~\ref{sec:appendix_DLV} on page~\pageref{sec:appendix_DLV}.
\\

Essentially the statement of this section is that in order to simulate a smooth arbitrage-free
option surface, fix (or also simulate) increasing $0<V_1<\cdots<V_M$ and
a $(N-2)\times M$ grid of discrete non-negative local volatilities~$\Sigma^i_j$. They have
the natural scale of annualized volatilities (indeed, as the grid density
increases the discrete local volatilities approach Dupire's continuous time
local volatility). The only constraint on the discrete local volatilities
is non-negativity, making them unique suitable for simulation and descriptive 
statistics. This is subject to further work.

\subsection{Smoothed Linear Interpolation within Bid/Ask}

In most of this article balance of focus is towards smoothness vs perfection of fit.
Indeed, the main purpose of our model is to ``fit'' a bigger number of market option prices with
a smooth model with fewer model strikes and expiries.
However, our approach can also be used for what we call ``smoothed linear interpolation'' within the bid/ask spread
following~\cite{cohen2020detecting}.

For this application we assume the model expiries match market expiries and that our (inner) model strikes $K^2_1,\ldots,K_M^{N_M-1}$
correspond to actual market strikes. Assume now that asks $A^i_j$ and bids $B^i_j$ for $j=1,\ldots,M$ and $i=2,\ldots,N_j-1$
have a non-zero spread. Assume further there exist arbitrage-free prices $C^i_j$
such that $B^i_j\leq C^i_j<A^i_j$ outside $i=0,1,N_j$. (We set $C^i_j:=(1-K^i_j)^+$ for $i=0,1,N_j$ in line with
theorem~\ref{th:noarb}.) We may find an arbitrage-free $C$ within bid/ask using linear programming.
Denote by $p$ the density implied by~$C$ via~\eqref{def:p}.

Similar to our production model, let
$$
    \hat C_j^i(\theta) := \sum_{\ell=1}^{N_j} p_j^\ell \mathrm{Call}\!\left( K_j^\ell, K^i_j, \theta T \right) \ .
$$
where $p_1,\ldots,p_M$ are the martingale densities computed from~$C$ via~\eqref{def:p}.
\begin{proposition}[Smoothed Linear Interpolation within Bid/Ask]\label{prop:interpol}
For all $\theta\geq 0$ we have $\hat C_i^\ell(\theta)\geq B_j^i$ for $i=2,\ldots,N_j-1$.

Moreover 
there exists a maximal $\theta^*>0$ such that $\hat C$ \emph{interpolates within bid/ask} in the sense that $B_j^i\leq \hat C_i^\ell(\theta^*)\leq A^i_j$
for $i=2,\ldots,N_j-1$.
\end{proposition}
\begin{proof}
This is a consequence of the fact that 
$
\hat C_j^i(\theta) \downarrow C^i_j \geq B^i_j
$
smoothly and monotonic as $\theta\downarrow 0$ for each pair $(j,i)$ with $i=2,\ldots,N_j-1$.
\end{proof}

\subsection{Generalization}\label{sec:genr}

The model~\eqref{eq:CasSumC} shines due to its simplicity. From a practical point
of view it is entirely sufficient for fitting observed market data within bid/ask
spreads.

For completeness we present now an extension which provides a more
advanced version which features a martingale~$Y$ per strike. We focus on 
a log-normal version for homogeneous strikes $0<K^1<\cdots<K^N$ for clarity, but other driving models~$Y$ can be utilized.
We view this model primarily as a conceptual result; in practice the production model already captures observed surfaces to high accuracy.

\begin{theorem}[Iterative Smooth Call Prices]\label{th:fullsmooth}
    Let $dV_j^i\geq 0$ be forward Black-Scholes implied variances for $i=1,\ldots,N$
    and $j=1,\ldots,M$. Let $q_1,\ldots,q_M$ be marginal densities
    with unit mean (e.g.~they do not have
    to be martingale densities). Define iteratively
    \begin{equation}\label{eq:tildeC1}
        \tilde c_1(K;v) := \sum_{i_1=1}^N q_1^{i_1}\, \mathrm{Call}\big( K^i, K;\,dV_1^i + v\big)
    \end{equation}
    and
    \begin{equation}\label{eq:tildeC1j}
        \tilde c_j(K; v) := \sum_{i_j=1}^N q_j^{i_j}\, K^{i_j}\ \tilde c_{j-1}\!\left(\frac{K}{K^{i_j}}; dV^{i_j}_j + v\right) 
    \end{equation}
    and then
    \begin{equation}\label{eq:tildeCTK}
        \tilde C_j(K) := \tilde c_{j+1}(K,0) \ \ \ \mbox{and} \ \ \
        \tilde C(T,K) := \alpha_j(T) \tilde C_{j+1}(K) + (1-\alpha_j(T))\, \tilde C_j(K) 
    \end{equation}
    with~$\alpha$ such as defined in~\eqref{eq:intT}.

    Then, $\tilde C$ is a smooth arbitrage-free option surface.    
\end{theorem}
\noindent
The proof is trivial observing that this pricing scheme 
is realized by
\begin{equation}\label{eq:Zcomplex}
    Z_j := \prod_{\ell=1}^j X_\ell Y_\ell^{X_\ell}
\end{equation}
with $X_\ell\sim q_\ell$
and $Y_\ell^i := \exp\left( \sqrt{ dV_\ell^i } N_\ell^i - \mbox{$\frac12$} dV_\ell^i \right)$
where $N_j^i$ is iid standard normal.

\begin{proposition}
    We may write
    \begin{equation}\label{eq:complexmodelit}
        \tilde C_j(K) := \sum_{i_j=1}^N q_j^{i_j}\,
        \left( \sum_{i_{j-1},\ldots,i_1} q^{i_{j-1}}_{j-1} \cdots q^{i_1}_1
        \ \mathrm{Call}\!\left(K^{i_j} \cdots K^{i_1},\, K;\, dV^{i_j}_j+\cdots dV^{i_1}_1 \right) 
        \right) \ .
    \end{equation}
\end{proposition}
\noindent
This formula is used in our implementation in section~\ref{sec:fullyconsistent}
on page~\pageref{sec:fullyconsistent} to iteratively fit the model to 
market data from the earliest expiry towards the latest expiry (the implementation allows
for different strikes per expiry; we omitted this here for sake of notational clarity).

As pleasing as this result is theoretically it is numerically more more
expensive to implement than our simple~\eqref{eq:CasSumC}
as we need to calculate the full tensor structure on the right hand side
of~\eqref{eq:complexmodelit} while iterating forward.

\begin{remark}[Extension to Martingale Densities]
    Equation~\eqref{eq:Zcomplex} shows that our representation means that the
    jump process~$X_j$ is independent of~$X_{j-1}$. We assumed this for computational
    efficiency, but it is noteworthy that we can also apply our idea to more general
    martingales~$X$ as follows: assume $q_1,\ldots,q_M$ is
    a martingale density with transition operators~$Q_{j|j-1}$. Define then
    \begin{equation}\label{eq:complexmodelitcond}
        \breve C_j(K) := \sum_{i_j=1}^N \,K^{i_j}
        \left( \sum_{i_{j-1},\ldots,i_1} Q_{j|j-1}^{i_j|i_{j-1}} Q_{j-1|j-2}^{i_{j-1}|i_{j-2}} \cdots Q^{i_1|0}_{1|0}
        \ \mathrm{Call}\!\left( K^{i_j},\,K;\, dV^{i_j}_j+\cdots dV^{i_1}_1 \right) 
        \right) \ .
    \end{equation}
    Then,~$\breve C$ has the martingale representation
    $$
        Z_j = X_j \prod_{\ell=1}^j Y_\ell^{X_\ell} \ .
    $$
\end{remark}    

\noindent
It is notable that this approach has a similar expression in a classic diffusion setting:

\begin{remark}[Generalization to a Diffusion Setting] Assume that
$$
    \frac{ dY^x_t }{ Y^x_t} = u_t(x) dW_t 
    \ \ \ \mbox{and} \ \ \ 
    \frac{ d X_t }{ X_t } = v(X_t) dB_t 
$$
are both positive martingales; then
$$
    \frac{ dZ_t }{ Z_t } = u_t(X_t) dW_t
$$
is a well-defined martingale. We then have
$$
     d(X_t Z_t)= X_t u_t(X_t) Z_t dW_t + Z_t v(X_t) X_t dB_t = (X_t Z_t) \left( u_t(X_t)dW_t + v(X_t)dB_t\right)
$$
which shows that $XY$ is also a positive martingale (albeit not Markov,
and correlated to~$X$). This idea can be trivially extended to processes with jumps.
\end{remark}

\section{Implementation}\label{sec:implementation}

We now discuss the practical implementation of fitting our model variants to observed market data. We will start with the
simplest variant and work our way to the fully consistent model of theorem~\ref{th:fullsmooth}. For both our production
and generalized model we provide the ability to specify different strikes per expiry.

\subsection{Homogeneous Strikes - Toy Model}

 To illustrate the general mechanics we start by providing a very simple algorithm for
the much simpler problem of fitting a model with homogeneous strikes to a market with the same
expiries and strikes. This is meant to be explain the basic ideas as it does not suffer from
too many indices.
In the the following section~\ref{sec:PROD} we present a more practical 
approach which handles inhomogeneous strikes etc. The reader may skip directly to that section.

Assume that we have strikes $0<K^1<\cdots<1<\cdots<N$ and expiries $0<T_1<\cdots<T_M$
with observed market mid-prices $C^i_j$. We also assume we have specified
 Black-Scholes implied variances $V_j$ which are increasing in expiries, $V_j \geq V_{j-1}$ and let~$\eta\in[0,1)$ be a variance factor, e.g.~$\eta=0.25$.
Define the matrix $\mathbf{C}_j\in \R^{ (N-2) \times N}$ for $j=1,\ldots,M$ as
\begin{equation}\label{eq:Cdef}
    \mathbf{C}_j^{\ell,i} := \mathrm{Call}( K^i, K^\ell, \eta V_j ) \ .
\end{equation}
for $\ell=2,\ldots,N-1$ and $i,\ell=1,\ldots,N$. This matrix maps the our model prices given as a linear function of our density $q_j$ to market prices.
Let also
$$
     \mathbf{U}^{\ell,i} := ( K^i - K^\ell )^+ 
$$
which are the maps from~$q$ to discrete call prices of the $q$'s themselves. We will use this to ensure
that the~$q$'s are increasing in call prices.

The candidate call prices based on $m$ for strikes $K$ at expiry $T_j$ are per~\eqref{eq:prod_model} given in terms of a density~$q_j$ as
$$
     \mathbf{c}_j := \mathbf{C}_j \cdot  q_j  \in \R^{N-2} 
$$
where as before ``$\cdot$'' denotes the classic matrix/vector product.

The fitting problem for our model can then be written as the linear program
\begin{equation}\label{eq:prob_homog}
\left\{
    \begin{array}{lll}
        \mbox{\textsl{Variables:}}
        \\
        \hspace{0.1cm}q_1,\ldots,q_M\ \mbox{with} \ q_j \in \R^{N} & \mbox{Marginal densities}\\
        \\
        \mbox{\textsl{Slack variables:}}
        \\
        \hspace{0.1cm}\mathbf{c}_j := \mathbf{C}_j \cdot q_j  \in \R^{N-2}  & \mbox{Prices at market strikes}\\
        \hspace{0.1cm} \mathbf{u}_j :=  \mathbf{U}_j \cdot q_j  \in \R^{N} &
        \mbox{Call prices for~$q_j$}\\
        \\
        \mbox{\textsl{Constraints:}}\\
           \hspace{0.1cm} q_j \geq 0 , \ 1'\cdot q_j = 1,  \ K'\cdot q_j = 1 & \mbox{Marginal 
            density with unit mean}\\
           \hspace{0.1cm}\mathbf{u}_j  \geq \mathbf{u}_{j-1}
        & \mbox{Martingale condition} \\  \\
        \mbox{\textsl{Objective:}} \\
        \hspace{0.1cm}\inf q_1,\ldots,q_M:\  \sum_j \left| w_j \cdot (C_j -\mathbf{c}_j)  \right|  
        & \mbox{Weighted fit to market}
    \end{array}
    \right.  \tag{SMP}
\end{equation}
Our example implementation we provided computes the tensor expressions in every iteration. This
means the algorithm has quadratic execution time and is in its current form not suitable
for production use.

\subsection{Production Model for Inhomogeneous Grids}\label{sec:PROD}

In practice the strikes in the market are not constant in time, in particular not when normalized by the forward
into ``pure'' strikes. When using real market data, it is also sometimes
necessary to add additional strikes when market strikes are far away from each other
(e.g.~after filtering by minimum quoted volumes). The code we used to test our
model adds strikes whenever market strikes are further apart than some maximum~$dx$,
and adds additional strikes outside the observed strike range.

Therefore, assume first that we are given boundary strikes $0<K^{\min}<1<K^{\max}$ which lie
outside any observed market prices. Assume further that we are given: 

\begin{itemize}
\item 
    \textbf{Market Expiries} $0<\tau_1<\cdots<\tau_m$.
\item 
    \textbf{Market Strikes}:
    For each market expiry we are given $n_\ell$ strikes $k^1_\ell,\ldots,k^{n_\ell}_\ell$ with $0<K^{\min}\ll k_\ell^1<\cdots<1<\cdots <
    k_\ell^{n_j} \ll K^{\max}$.

\item
    \textbf{Market Prices}:
    For each $k^r_\ell$ we are given a target price $C_\ell^r$, an ask price $A_\ell^r$
    and a bid prices $B_\ell^r\geq 0$. We assume that the spread is strictly positive: $A_\ell^r>B^r_\ell$.

\item 
    \textbf{Weights}: For each option we are given a weight $w^r_\ell\in\R_{\geq 0}$,  typically the inverse of the prevailing bid/ask spread,  
    or the inverse of Vega to approximate a fit in implied volatilities.

\item
    \textbf{Market ATM Variances}: we assume we are given market ATM log-normal variances
    $0<v_1<\cdots<v_m$.
\end{itemize}

\noindent
The model is itself defined as before:
\begin{itemize}
\item 
    \textbf{Model Expiries}: $0<T_1<\cdots<T_M$.

\item  
    \textbf{Model Strikes}: Each model expiry has $N_j$ strikes $0=K^{\min}=K^1_j<\cdots<1<\cdots<K^{N_j}_j=K^{\max}$
    which may or may not include
    any market strikes.

\item  
    \textbf{Implied Variances} we also assume we have an
    increasing set of implied variances $0<V_1<\cdots<V_m$
    and associated scaling factor~$\eta\in[0,1)$, for example
    $\eta=0.25$. If $\eta$ is zero, then the model becomes linear. Figure~\ref{fig:eta}
    in page~\pageref{fig:eta} illustrates
    the impact of~$\eta$.
\end{itemize}

As a first step we identify the location of the market expiries with respect to the model expiries.
To this end, we interpolate model prices using some weighting function~$\alpha$ as in~\eqref{eq:intT}.
That means that the model prices for expiry~$\tau_\ell$ can be written
as convex combination of prices at the two surrounding model expiries $T_{j(\ell)-1}<\tau_\ell\leq T_{j(\ell)}$ as
$$
    C^r_\ell = \alpha_\ell\, \mathrm{Call}( K_{j(\ell)},k^r_\ell,\eta v_\ell)' \cdot q_{j(\ell)}
             + (1-\alpha_\ell)\,  \mathrm{Call}( K_{j(\ell)-1},k^r_\ell,\eta v_\ell)' \cdot q_{j(\ell)-1}
$$
with weights~$\alpha_\ell := \alpha_{j(\ell)}(\tau_\ell)$
if $V_{j(\ell)-1}\leq v_\ell \leq V_{j(\ell)}$.
We therefore pre-compute matrices $\mathbf{C}_{\ell|+} \in \R^{n_{j(\ell)}\times N_j}$
and $\mathbf{C}_{\ell|-} \in \R^{n_{j(\ell)-1}\times N_j}$
as
\begin{equation}\label{eq:imp_matrixes_market}
\begin{array}{llrl}
    \mathbf{C}_{\ell|+}^{r,i} & := & \alpha_\ell & \mathrm{Call}( K^i_{j(\ell)}\, k^r_\ell, \eta v_\ell )  \ \mbox{and} \\
    \mathbf{C}_{\ell|-}^{r,i} & := & (1-\alpha_\ell ) & \mathrm{Call}( K^i_{j(\ell)-1}\, k^r_\ell, \eta v_\ell )  \ . \\
\end{array}
\end{equation}

In order to ensure that call prices are increasing in time 
we also introduce the matrices $\mathbf{U}_j\in\R^{N_j,N_j}$ and $\mathbf{R}\in\R^{N_{j+1},N_j}$
given as
\begin{equation}\label{eq:imp_matrixes}
\begin{array}{lll}
    \mathbf{U}_j^{\ell,i} & := & \mathrm{Call}( K^i_j, K^\ell_j, \eta\omega V_j ) \ \mbox{and} \\
    \mathbf{R}_j^{\ell,i} & := & \mathrm{Call}( K^i_{j-1}, K^\ell_j, \eta\omega V_{j-1} ) \ .
\end{array}
\end{equation}
The variable~$w\in\{0,1\}$ is explained below.

Our linear program then becomes:
\begin{equation}\label{eq:HMOG}
\left\{
    \begin{array}{lll}
        \mbox{\textsl{Variables:}}
        \\
        \hspace{0.1cm}q_1,\ldots,q_M\ \mbox{with} \ q_j \in \R^{N_j} & \mbox{Marginal densities}\\
        \\
        \mbox{\textsl{Slack variables for $\ell=1,\ldots,m$:}}\\
        \hspace{0.1cm}\mathbf{c}_\ell := \mathbf{C}_{\ell|+} \cdot q_{j(\ell)} + \mathbf{C}_{\ell|-} \cdot q_{j(\ell)-1} \in \R^{n_j} & \mbox{Prices at market strikes}\\
        \\
        \mbox{\textsl{Slack variables for $j=1,\ldots,M$:}}\\
        \hspace{0.1cm}\mathbf{u}_j :=  \mathbf{U}_j \cdot q_j  \in \R^{N_j}  & \mbox{Prices at model strikes}\\
        \hspace{0.1cm}\mathbf{r}_j := \mathbf{R}_j \cdot q_{j-1} \in \R^{N_j} & \mbox{Previous model prices at current strikes}\\
        \\
        \mbox{\textsl{Constraints:}} \\
            \hspace{0.1cm}q_j \geq 0 , \ 1'\cdot q_j = 1,  \ K'\cdot q_j = 1 & \mbox{Marginal 
            density with unit mean}\\
            \hspace{0.1cm}\mathbf{u}_j \geq \mathbf{r}_j & \mbox{Martingale 
            condition $(*)$}\\ \\
        \mbox{\textsl{Objective:}} \\
        \hspace{0.1cm}\inf q_1,\ldots,q_M:\ \sum_\ell \left| w_\ell\cdot (C_\ell -\mathbf{c}_\ell)  \right| 
        & \mbox{Weighted fit to mid}
    \end{array}
    \right. \tag{MDL}
\end{equation}

In order to ensure absence of arbitrage in time condition~$(*)$ must be
satisfied with $\omega=0$ as this ensures that the discrete strike marginal
densities constitute a martingale density, c.f.~theorem~\ref{th:noarb}
on page~\pageref{th:noarb}.
However, as noted there this is in fact just a sufficient but not necessary
condition for absence of arbitrage. From a practical and numerical perspective it is also more
natural to impose the price condition in observed price space, which
is the case for~$w=1$. However, this does then not guarantee strict absence of arbitrage.

\begin{remark}[Bid/Ask Spreads]
    Actual markets do not have a mid-price to fit to; instead we observe bid $B_j^i$ and
    ask prices $A_j^i$ which are a positive spread~$A_j^i - B_j^i$ apart.

    To incorporate bid/ask's, there are a number of trivial approaches to account for bid/ask spreads.
    \begin{enumerate}
        \item Normalize the fitting error by scaling by the observed bid/ask spread i.e.
        \begin{equation}\label{eq:bidaskweight}
            w^i_j := \frac1{ A^i_j - B^i_j } \ .
        \end{equation}
        
        \item enforce the fitted prices to be within bid/ask by adding the constraints $B_j \leq C_j \leq A_j$. This yields the additional linear constraints
        $$
        \left\{
            \begin{array}{lll}
            \mbox{\textsl{Additional Constraints:}} \\
            \hspace{0.1cm}\max\{ C_j - A_j, 0 \} = 0 \\
            \hspace{0.1cm}\max\{ B_j - C_j, 0 \} = 0 \ . \\
            \end{array}\right.
        $$

        \item Increase the fitting penalty when the fitted price is outside bid/ask.
        $$
        \left\{
            \begin{array}{lll}
            \mbox{\textsl{Alternative Objective:}} \\
             \hspace{0.1cm}\inf_{q_1,\ldots,q_M}:\ \sum_j  w_j \Big( 
                \epsilon | C_j - \mathbf{c}_j | + \max\{ C_j - A_j, 0 \} + \max\{ B_j - C_j, 0 \} \Big) \ .
            \end{array}\right.
        $$
        This is our recommended default setting with~$\epsilon=10^{-8}$ with
        using weights~\eqref{eq:bidaskweight}.
        \end{enumerate}
    All of these fit into a linear or quadratic programming framework. 
\end{remark}

\begin{figure}[H]
    \centering
    \includegraphics[width=0.9\linewidth]{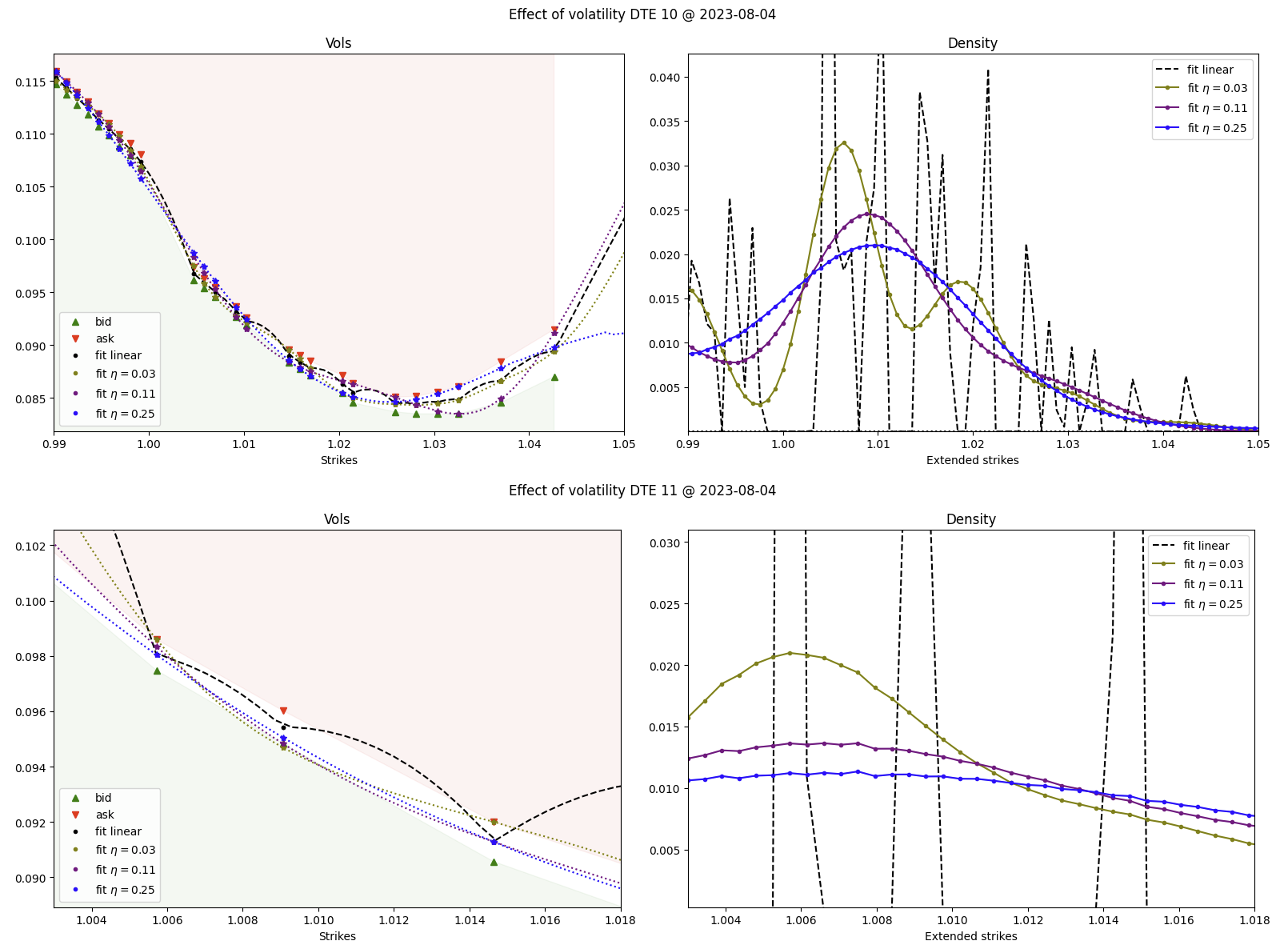}
    \caption{These two examples illustrate the impact of the smoothness parameter~$\eta$ 
    in a zoomed in view on a fit. The left hand graphs show the implied
    volatility bid/ask, and the model fits. The right hand graph shows the density produced by each model setting. ``Linear'' refers to the case~$\eta=0$.
    The implied variances~$V_j$ used in the model are ATM bid implied volatilities.
    The lower example in particular illustrates well the balancing act between the
    usefulness of a smooth surface vs a good fit: the left hand volatility
    graph show that the black dashed line from the linear model fits perfectly at the mid point within bid/ask at market strikes, as expected.
    However, the dashed line also shows the effect of linear interpolation
    being ``too expensive'' between market strikes as the implied volatilities bump upwards.
    In contrast, the models with non-zero smoothness provide much more natural fits 
    the market. The smoothest model has some fitting error, but also provides
    a~more natural interpolation across strikes. 
    }
    \label{fig:eta}
\end{figure}

\subsection{Generalized Model}\label{sec:fullyconsistent}

The final step is to implement fitting for our generalized model as described in theorem~\ref{th:fullsmooth}.
To lighten notation we assume that model expiries match market expiries; an extension in the spirit
of the previous section is straight forward.
Equation \eqref{eq:Zcomplex}
on page~\pageref{eq:Zcomplex}
can be written as
$$
    \tilde C_j(K) := \sum_{i_j=1}^{N_j} q_j^{i_j} \ \left( 
            \sum_{i_{j-1},\ldots,i_1=1}^{N_{j-1},\ldots,N_1}
                q_{j-1}^{i_{j-1}} \cdots q_1^{i_1}\  \cdots K^{i_1}_1\ \mathrm{Call}\left( K^{i_j}_j\,K^{i_{j-1}}_{j-1}\, ;\ 
                    K^\ell_j\, ;
                \  dV^{i_j}_j+\cdots+dV^{i_1}_1 \right) 
            \right) \ .
$$
That means that in addition to the input assumptions of the previous section 
we will also assume that for each $j$ there is a vector~$dV^1_j,\ldots,dV^{N_j}_j$
of model incremental variances each corresponding to one of
the log-normal increments~$Y^i_j$ in theorem~\ref{th:fullsmooth}.

We will now present an iterative fitting scheme which fits each step using linear programming
going forward along expiries.

Define the tensors
$$
\begin{array}{llllll}
    {\mathbb V}_j & := & dV_j \oplus \cdots \oplus dV_1 & = & ( dV^{i_j}_j+\cdots+dV^{i_1}_1 )_{i_j,\ldots,j_1} & \in \R^{N_j\times \cdots \times N_1}\\
    \K_j & := & K_j \otimes \cdots \otimes K_1 & = &  ( K^{i_j}_j \cdots K^{i_1}_1 )_{i_j,\ldots,j_1} & \in \R^{N_j\times \cdots \times N_1}\\
    \Q_j & := & q_j \otimes \cdots \otimes q_1 & = & ( q^{i_j}_j \cdots q^{i_1}_1 )_{i_j,\ldots,j_1} & \in \R^{N_j\times \cdots \times N_1}
\end{array}
$$
and then
$$
\begin{array}{llll}
    \mathbf{C}_j 
    & := & \mathrm{Call}\left( \K_j,\, k_j,\,\eta {\mathbb V}_j \right) & \in \R^{n_j \times N_j \times \cdots \times N_1}\\
    \mathbf{U}_j 
    & := & \left( \K_j - K_j \right)^+ & \in \R^{N_j \times N_j \times \cdots \times N_1}\\
    \mathbf{R}_j 
    & := & \left( \K_{j-1} - K_j \right)^+ & \in \R^{N_j \times N_{j-1} \times \cdots \times N_1}
\end{array}
$$
where indexing is as implied by the dimensionality of each tensor.

We note that provided~$\Q_j$ the inner product
$$
    \mathbf{c}_j := \mathbf{C}_j \cdot \Q_j' \in \R^{n_j}
$$
gives the model prices for the market strike at the $j$th expiry.
However,~$\Q_j$ contains products of all the vectors~$q_j$ such that above is no longer
linear or quadratic.\\

\noindent
\textbf{Iterative Linear Programming:} we therefore propose the following scheme: assume that 
we have found $q_1,\ldots,q_{j-1}$ and therefore $\Q_{j-1}$, and that we are now 
looking to fit the marginal density for
the $j$th expiry, $q_j\in [0,1]^{N_j}$.

\begin{equation}\label{eq:GENR}
\left\{
    \begin{array}{lll}
        \mbox{\textsl{Variables:}}
        \\
        \hspace{0.1cm}q_j \in \R^{N_j} & \mbox{Marginal density}\\
        \\
        \mbox{\textsl{Slack variables:}}\\
        \hspace{0.1cm}\mathbf{c}_j := ( \mathbf{C}_j \cdot \Q_{j-1}' )\cdot q_j \in \R^{n_j} & \mbox{Prices at market strikes}\\
        \hspace{0.1cm}\mathbf{u}_j := ( \mathbf{U}_j \cdot \Q_{j-1}' )\cdot q_j  \in \R^{N_j}  & \mbox{Prices at model strikes}\\
        \hspace{0.1cm}\mathbf{r}_j := ( \mathbf{R}_j \cdot \Q_{j-1}' ) \in \R^{N_j} & \mbox{Previous model prices at current strikes}\\
        \\
        \mbox{\textsl{Constraints:}} \\
            \hspace{0.1cm}q_j \geq 0 , \ 1'\cdot q_j = 1,  \ K'\cdot q_j = 1 & \mbox{Marginal 
            density with unit mean}\\
            \hspace{0.1cm}\mathbf{u}_j \geq \mathbf{r}_j & \mbox{Martingale 
            condition}\\ \\
        \mbox{\textsl{Objective:}} \\
        \hspace{0.1cm}\inf q_j:\ \left| w_j\cdot (C_j -\mathbf{c}_j)  \right| 
        & \mbox{Weighted fit to market}
    \end{array}
    \right. \tag{GNR}
\end{equation}

\section{Example}

The following figures~\ref{fig:fullfit}ff illustrate the impressive performance of our algorithm on real data 
sourced from Option Metric's Ivy DB data based via WRDS:
we fitted 1000 options within 2 ATM implied volatility standard deviations which had a $\mathrm{Vega}\sqrt{T}$ of at least~0.1\%
on 2025-05-06 across all 48 expiries from~1 to~657 business days. Options were chosen by closeness to ATM. Time to expiry
for AM settled options is computed as business days until the day before plus~20\%.
The model fitted 91.4\% of all options within bid/ask.
Of those options not fitted the median error is just 21\% of half spread. The fit took sub-one second on a desktop PC.

\begin{figure}[H]
    \includegraphics[width=0.9\linewidth]{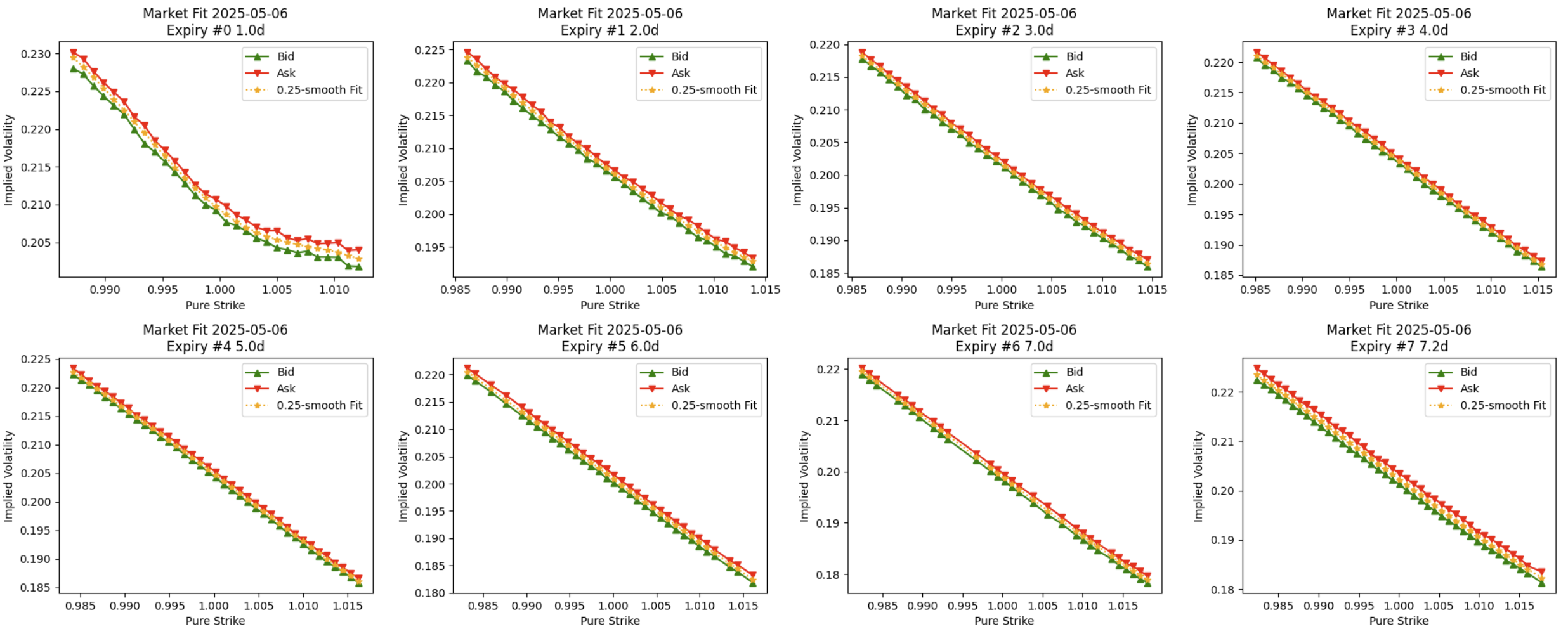}
    \caption{Fit to SPX on 2025-05-06 (1/6).
    \label{fig:fullfit}
    }
\end{figure}
\begin{figure}[H]
    \includegraphics[width=0.9\linewidth]{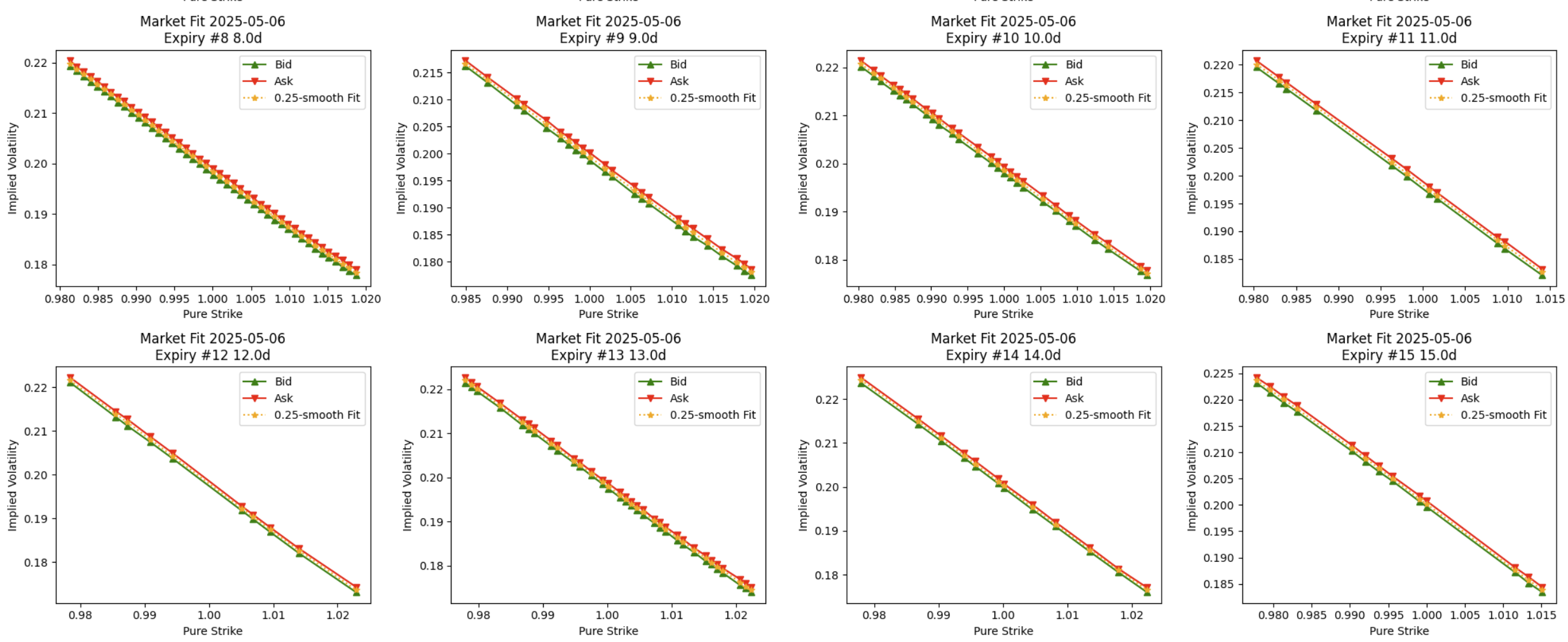}
    \caption{Fit to SPX on 2025-05-06  (2/6).
    }
\end{figure}
\begin{figure}[H]
    \includegraphics[width=0.9\linewidth]{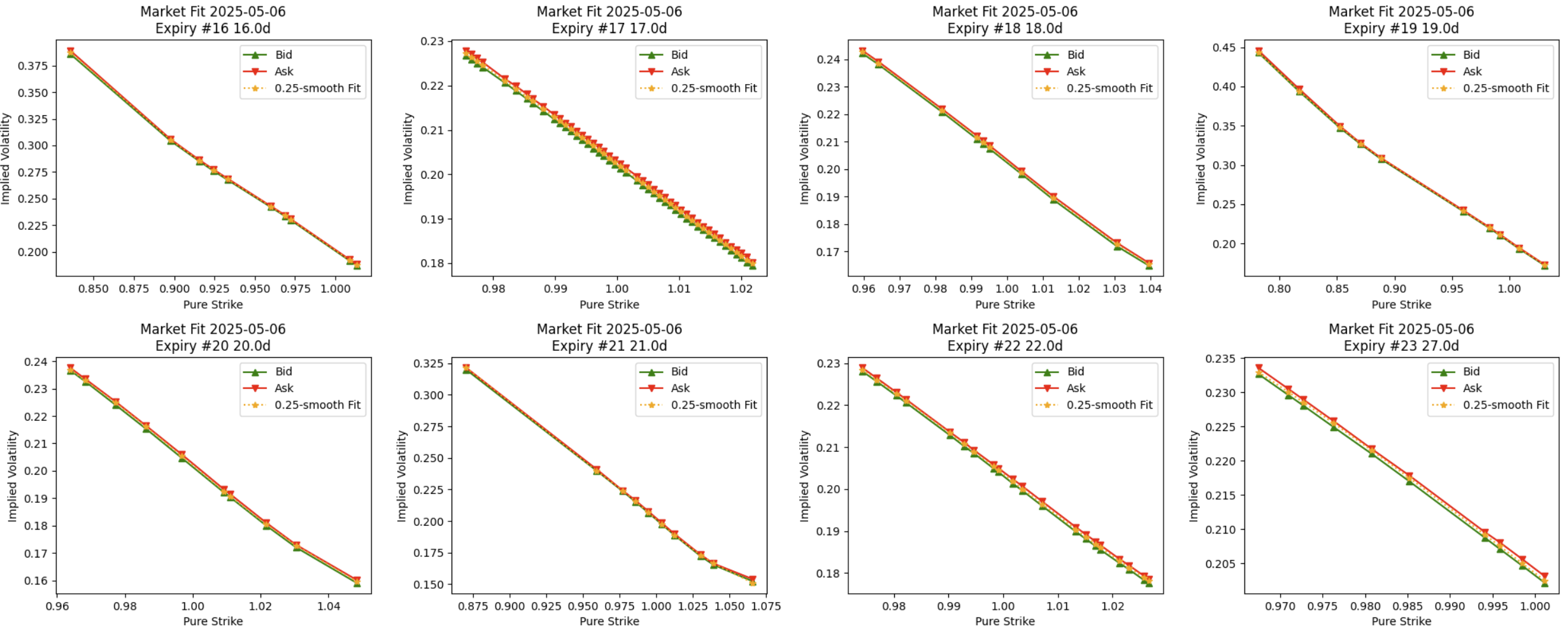}
    \caption{Fit to SPX on 2025-05-06  (3/6).
    }
\end{figure}
\begin{figure}[H]
    \includegraphics[width=0.9\linewidth]{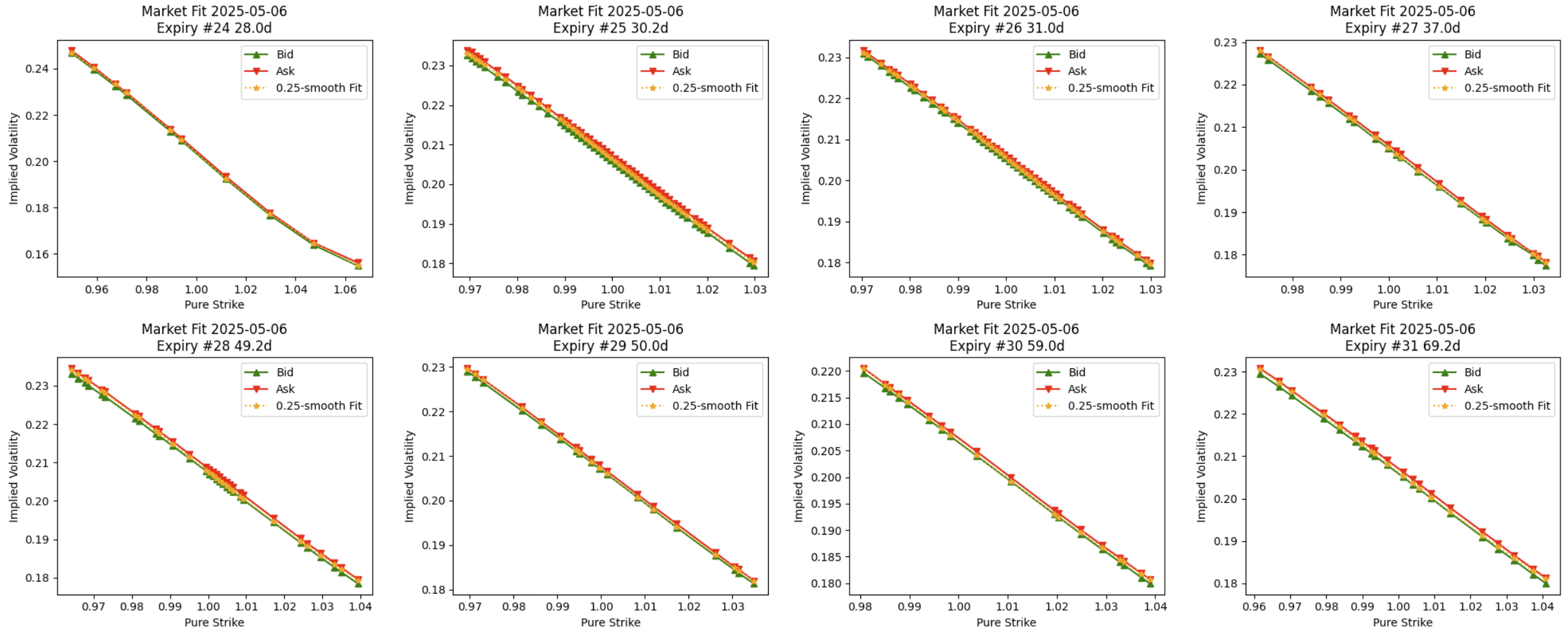}
    \caption{Fit to SPX on 2025-05-06  (4/6).
    }
\end{figure}
\begin{figure}[H]
    \includegraphics[width=0.9\linewidth]{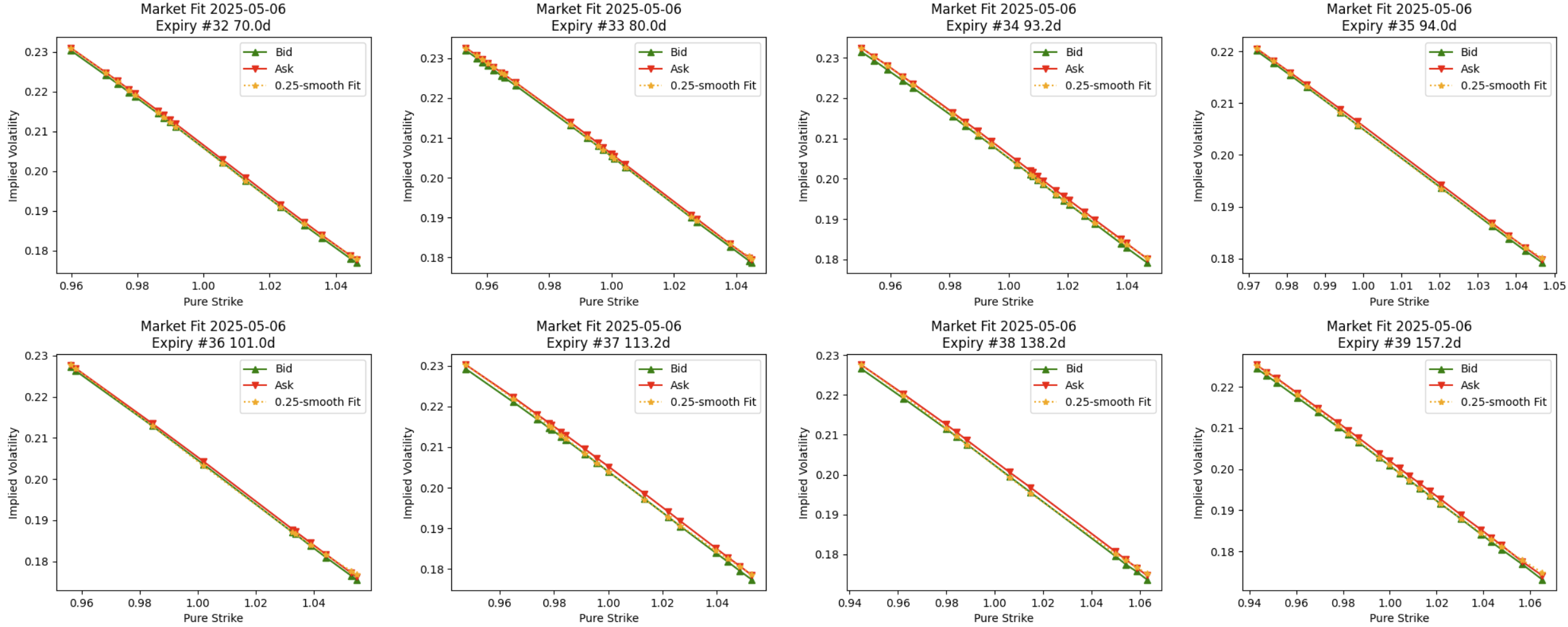}
    \caption{Fit to SPX on 2025-05-06  (5/6).
    }
\end{figure}
\begin{figure}[H]
    \includegraphics[width=0.9\linewidth]{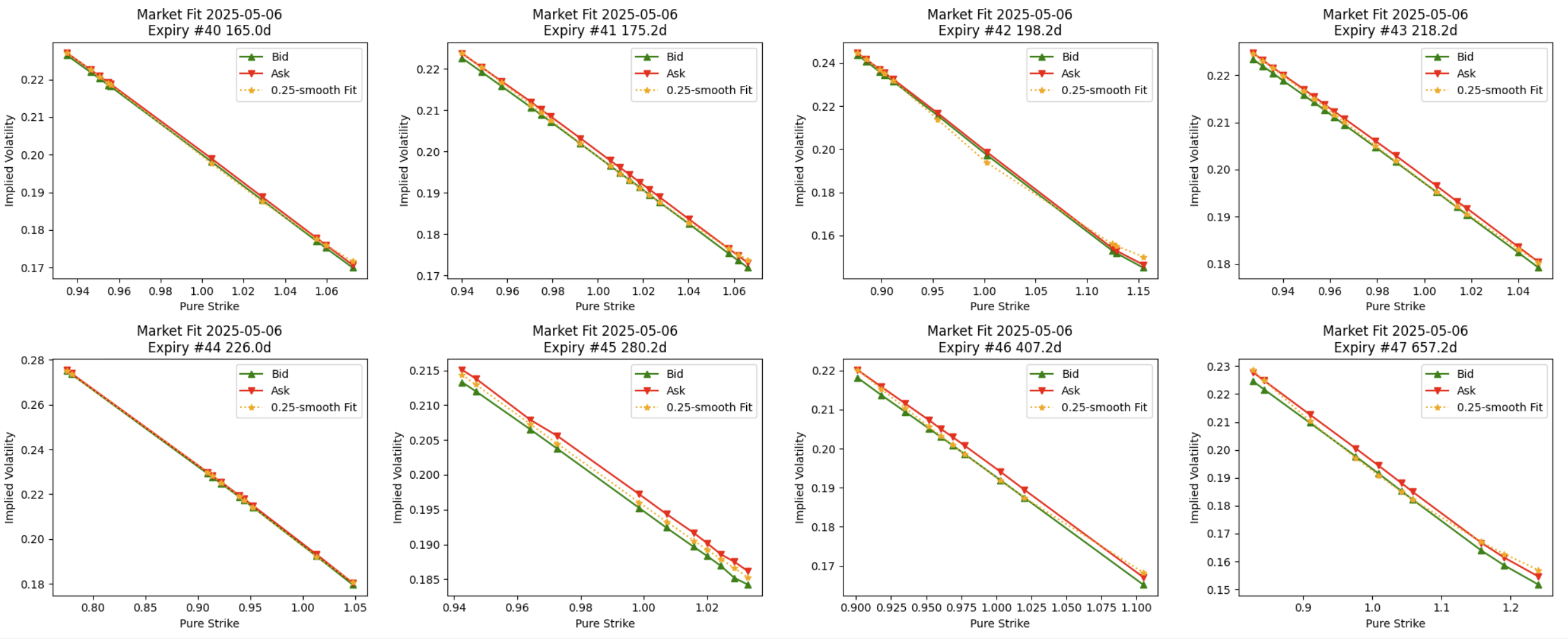}
    \caption{Fit to SPX on 2025-05-06  (6/6).
    }
\end{figure}

\section{Conclusion}

We have presented a smooth strictly arbitrage-free parametrization of an option price surface and shown how to efficiently fit
such model -- and its simpler variant -- to market data, illustrating the broad applicability of the method. We have also shown how
to incorporate bid/ask spread into our fit. 
Our models can be written in terms of ``volatility'' parameters
which only have to be positive in order to make the model smooth and strictly arbitrage-free. We presented an interpretation
as a dynamic process to aid intuition behind our approach.

Our approach is a natural extension of linear option surface fitting whenever smoothness is critical to the application.

\section*{Funding}
A.\ Kratsios acknowledges financial support from an NSERC Discovery Grant No.\ RGPIN-2023-04482 and No.\ DGECR-2023-00230.   and by the project Bando PRIN 2022 named ``Qnt4Green - Quantitative Approaches for Green Bond Market: Risk Assessment, Agency Problems and Policy Incentives'', codice 2022JRY7EF, CUP E53D23006330006, funded by European Union – NextGenerationEU, M4c2.
A.\ Kratsios and R. Saqur were also supported, in part, by the Province of Ontario, the Government of Canada through CIFAR, and companies sponsoring the Vector Institute\footnote{\href{https://vectorinstitute.ai/partnerships/current-partners/}{https://vectorinstitute.ai/partnerships/current-partners/}}.

\appendix

\section{Appendix}

\subsection{Absence of Arbitrage}\label{sec:appendix_absence}

We present a self-contained proof for theorem~\ref{th:contarbfree}  on page~\pageref{th:contarbfree}
which we restate here for convenience:

\begin{theorem}[Theorem \ref{th:contarbfree}]
    A~continuous call price surface $C:[0,\infty)^2 \rightarrow [0,\infty)$ \emph{arbitrage-free} 
    if and only if
    \begin{enumerate}
        \item The market has unit expectation: $C(T,0)\equiv \E[Z_T]=1$ for all $T$;
        \item zero is unattainable: $\partial_K C(T,0)\equiv -1$ since then $\P[ Z_T=0 ] =0$;  
        \item call prices ultimately reach zero: $\lim_{K\uparrow \infty} C(T,K)=0$,
        \item call prices are convex: $\partial_{KK} C(T,K) \geq 0$; and
        \item call prices are increasing in time: $\partial_T C(T,K)\geq 0$.
    \end{enumerate}
\end{theorem}

\noindent
To prove this statement, we first recall per expiry the following proposition~2.3 from~\cite{buehler2006expensive}::

\begin{proposition}\label{prop:arbfreemart}
A function $c:[0,\infty)\mapsto[0,1]$ is arbitrage-free in the sense that there exist a random variable~$Z$
with unit mean such that $c(K)=\E[ (Z-K)^+]$ if and only if
\begin{enumerate}
    \item $c$ is decreasing,
    \item $c$ is convex, and
    \item $c(0)=1$, \label{it:noarbnbo}
$\lim_{K\uparrow \infty} c(K) = 0$ and
                              $c'(0):=\lim_{\delta\downarrow 0} \mbox{$\frac1\delta$} c(\delta)\geq -1$.
\end{enumerate}
Moreover, if $c'(0)=-1$ then $Z$ is positive.
\end{proposition}
\begin{proof}
    Since $c$ is convex and decreasing, its right-hand derivative $c'$
    exists and is right-continuous and non-decreasing with $c'\geq -1$. Since $c(0)=1$ and $\lim_{K\uparrow \infty} c(K)=0$
    we also conclude that $c$ is non-negative.
    
    Let $f(K) := 1 + c'(K)$, which is
    a positive, right-continuous and non-decreasing function and which
    therefore implies the existence of some positive $\sigma$-finite measure via
    $\tilde \nu[(a,b]]:=f(b)-f(a)$ cf.~\cite{Aliprantis99InfiniteDimAnalysis} theorem~9.47 pg.\,354.
    The properties $\lim_{K\uparrow\infty} c(K) = 0$ and $c(K) \geq 0$
    imply that $\lim_{K\uparrow\infty} c'(K) = 0$,
    hence we obtain that $\nu^*:=\tilde \nu[(0,\infty)]
    = \lim_{K\uparrow \infty} f(K) - f(0) = 0 - c'(0) \in [0,1]$.    
    The proof
    is complete by defining
    $
        \nu(A) := (1-\nu^*)\, \delta_{0}(A) + \tilde \nu(A)
        = (1+c'(0))\, \delta_{0}(A) + \tilde \nu(A)
    $
    for measurable sets $A \subseteq \R^+_0$. 
    This measure has unit expectation
    as a result of the assumption~$c(0)=1$.

    We note in particular that if and only if~$c'(0)=-1$ then~$\nu$ has mass only in the positive
    numbers, confirming positivity of~$Z$.
\end{proof}

\noindent
To account for term structure we refer to the following historical result from~\cite{Kellerer1972}:

\begin{theorem}[Kellerer 1972]\label{th:kellerer}
     Let $\mathcal{M}=(\mu_T)_{T\in\mathcal{T}}$ be a set of probability
     measures with unit expectation.
     Then, a Markov martingale $Z=(Z_T)_{T\in\mathcal{T}}$ with marginal distributions $\mu_T$
     exists if and only if $\mathcal{M}$ is in Balayage-order, that is
     \[
        \mu_T \preceq \mu_U
     \]
     for all $T,U\in\mathcal{T}$ with $T<U$ where $\preceq$ denotes convex order.
\end{theorem}

\begin{proof}[Proof of theorem~\ref{th:contarbfree}]
    In light of proposition~\ref{prop:arbfreemart} and
    theorem~\ref{th:kellerer} it remains to show that~\textit{\ref{it:caf_dt}.}~in theorem~\ref{th:contarbfree}
    implies that $\E[ f(Z_T) ] \leq \E[ f(Z_U) ]$ for any $T<U$ and all
    convex~$f$.

    For a convex function~$f$
    $$
        f(x) = f(0) + f'(0) x + \int_0^x (x-K)^+ \mu(dK)
    $$
    where $\mu[(a,b]] := f'(b)-f'(b)$ is the measure implied by~$f$. In a
    generalized sense\footnote{
        There is a countable number of extremal points (c.f.~\cite{buehler2006expensive}) $x_k$ such that
        $f''(K)\in\R_{\geq 0}$ for $K\not\in\{x_k\}_k$ and $f''(x_k)=p_k \delta_{x_k}$
        for $p_k\in\R_{\geq 0}$ and $\delta_{\cdot}$ denotes the Dirac measure.
    }this can be written as
    $$
        f(x) = f(0) + f'(0) x + \int_0^x (x-K)^+ f''(K) dK
    $$
    where~$ f''(K)\geq 0$.
    Taking expectations proves the statement.
\end{proof}

\subsection{Transition Operators with Discrete Local Volatilities}\label{sec:appendix_DLV} 

We present a proof for theorem~\ref{th:dlvtrs} on page~\pageref{th:dlvtrs}:
we start with the following statement which is essentially from~\cite{AndreasenHuge}:
\begin{theorem}[Construction of Transition Matrices]\label{th:zmatrix}
    Assume that~$M\in\R^{N\times N}$ is a square matrix with columns adding up to~1, and 
    all whose off-diagonal elements are non-positive.
    
    Then, its inverse exists, is non-negative, and its columns add up to~1; in other
    words~$M^{-1}\in\R^{N\times N}$ is a probability matrix.
    
    Moreover, if~$M$ has the martingale property with respect to~$K\in\R^N$
    in the sense that $K'M=K'$ then~$M^{-1}$ has the martingale property with respect to~$K$,
    too, in the sense that $K'=K'M^{-1}$.
\end{theorem}
\begin{proof}
A~matrix~$M$ with the above properties is a
Z-matrix, i.e.~$M^{i,i}>0$ and~$M^{i,\ell}\leq 0$ for~$i\not=\ell$. Moreover, 
its columns add up to 1, which means
it is a non-singular M-matrix, c.f.~\cite{alma995837143606533}, chapter~6:
equivalent classification $I_{29}$ of non-singular M-Matrices.

Hence, its inverse exists and is non-negative, c.f.~\cite{alma995837143606533}, chapter~6:
equivalent classification $N_{38}$ of non-singular M-Matrices.

Finally,~$1' M = 1'$ implies~$1'  = 1' M^{-1}$, i.e.~the columns of~$M^{-1}$
add up to~1, too. Similarly, if~$K' M = K'$, then~$K' = K' M^{-1}$.\
\end{proof}

\begin{proof}[Proof of theorem~\ref{th:dlvtrs} on page~\pageref{th:dlvtrs}]
We show that~$Q_{j|j-1}$ defined in~\eqref{eq:implicit} on page~\pageref{eq:implicit}
is indeed a transition operator.
\begin{itemize}
    \item The inverse satisfies the conditions of theorem~\eqref{th:zmatrix} hence
    $Q_{j|j-1}$ is a transition matrix.
    \item We now show $K_j'\cdot Q_{j|j-1}=K_j'$ by showing
    $K'=K'\cdot Q_{j|j-1}^{-1}$. Let
    $$
        U^i_j := \frac{ (\Sigma_{i,j} K^i)^2 (T_j-T_{j-1})}{ K^{i+1} - K^{i-1} } \ .
    $$
    such that
    $$
        w^{i+}_j = \frac{ U^i_j }{ K^{i+1} - K^i } 
        \ \ \ \mbox{and} \ \ \
        w^{i-}_j = \frac{ U^i_j }{ K^i - K^{i-1} }  \ .
    $$
    For $i=2,\ldots,N-1$ this yields
    \begin{eqnarray*}
        (K'\cdot Q_{j|j-1}^{-1})^i
        & = &
        -
        K^{i-1} \frac{ U^i_j }{ K^i - K^{i-1} }  
        +
        K^i + K^i \left( \frac{ U^i_j }{ K^i - K^{i-1} }   + \frac{ U^i_j }{ K^{i+1} - K^i } \right)
        \\
        &  &\qquad \qquad\qquad\qquad\qquad\qquad\qquad\qquad\qquad\qquad - K^{i+1} \frac{ U^i_j }{ K^{i+1} - K^i }  \\
        & = &
        K^i + 
         U^i_j \left(
         \frac{  K^i  - K^{i-1}}{ K^i - K^{i-1} }  +    
         \frac{ K^i- K^{i+1} }{ K^{i+1}{} - K^i }          
        \right)\ = \  K^i \ .
    \end{eqnarray*}
\end{itemize}
\end{proof}

\begin{proof}[Proof of theorem~\ref{th:lintransition} on page~\pageref{th:lintransition}]
\begin{eqnarray}
    1 L^{\cdot,i}_{j|j-1} & = &
            \sum_{\ell=1}^{N_j} 
        \left( 
        \frac{ ( K^i_{j-1} - K^{\ell+1}_j )^+ - ( K^i_{j-1} - K^\ell_j )^+ }{ K^{\ell+1}_j - K^{\ell}_j }
        -
        \frac{ ( K^i_{j-1} - K^{\ell}_j )^+ - ( K^i_{j-1} - K^{\ell-1}_j )^+ }{ K^{\ell}_j - K^{\ell-1}_j }
        \right)
        \\
        & = &
        \frac{ ( K^i_{j-1} - K^{N_j+1}_j )^+ - ( K^i_{j-1} - K^{N_j}_j )^+ }{ K^{N_j+1}_j - K^{N_j}_j }
        -
        \frac{ ( K^i_{j-1} - K^1_j )^+ - ( K^i_{j-1} - K^0_j )^+ }{ K^1_j - K^0_j }
        \\
        & = &
        0
        -
        \frac{ ( K^i_{j-1} - K^1_j )^+ - ( K^i_{j-1} - K^0_j )^+ }{ K^1_j - K^0_j } \ = \ 1 \ .
\end{eqnarray}
Define for each $K^i_{j-1}$ $l(i) := \min\{\ell=0,\ldots,N_j+1:\,K^i_{j-1} > K^\ell_j\}$
where we added a fictitious $K^{N_j+1}_j\gg K^{N_j}_j$. Then
\begin{eqnarray}
    K_j L^{\cdot,i}_{j|j-1} & = &
            \sum_{\ell=1}^{N_j} K_j^\ell
        \left( 
        \frac{ ( K^i_{j-1} - K^{\ell+1}_j )^+ - ( K^i_{j-1} - K^\ell_j )^+ }{ K^{\ell+1}_j - K^{\ell}_j }
        -
        \frac{ ( K^i_{j-1} - K^{\ell}_j )^+ - ( K^i_{j-1} - K^{\ell-1}_j )^+ }{ K^{\ell}_j - K^{\ell-1}_j }
        \right)
        \\
        & = &
            \sum_{\ell=0}^{N_j} ( K_j^\ell - K_j^{\ell+1} )  
        \frac{ ( K^i_{j-1} - K^{\ell+1}_j )^+ - ( K^i_{j-1} - K^\ell_j )^+ }{ K^{\ell+1}_j - K^{\ell}_j }
        \\
        & = &
            \sum_{\ell=0}^{N_j} 
        ( K^i_{j-1} - K^\ell_j )^+ - ( K^i_{j-1} - K^{\ell+1}_j )^+ 
        \\
        & = &
            \sum_{\ell=0}^{N_j} \1_{K^i_{j-1} \geq K^{\ell+1}_j } ( K^\ell_j - K^{\ell+1}_j ) 
            +  \sum_{\ell=0}^{N_j}  \1_{K^{\ell}_j < K^i_{j-1} \leq K^{\ell+1}_j } ( K^i_{j-1}  - K^\ell_j  )
        \\
        & = &
            K^{l(i)}_j
            + ( K^i_{j-1}  - K^{l(i)}_j  )  \ = \ K^i_{j-1} \ .
\end{eqnarray}
\end{proof}

\begin{proof}[Proof of theorem~\ref{th:dlvtoQ} on page~\pageref{th:dlvtoQ}]
This next step is to show that a discrete local volatility as defined in~\eqref{eq:dlv2}
reprices the input call prices. To do so
    we re-iterate the proof of section~4.3.2 in~\cite{DLV}
    and show that with $\Sigma$ as defined above $ Q_{j|j-1}^{-1} p_j = \bar p_{j-1}$. We first note that
    $$
        \Sigma^2_{i,j} := \frac{2}{(K^i)^2} \left( \frac{ C_j^i - \bar C_{j-1}^i }{ T_j - T_{j-1} } \right)/\left( \frac{ p^i_j}{ \frac12 (K^{i+1}_j-K_j^{i-1}) } \right)
    $$
    hence $w$ defined in~\eqref{eq:wDLV} can be written as~\eqref{eq:wself}.
    For $i\in\{2,\ldots,N-1\}$ we therefore get
    \begin{align*}
        ( Q_{j|j-1}^{-1}\cdot p_j  )^i
        & = 
            - p_j^{i-1}  w^{(i-1)+}_j
            +
            p_j^{i} (1 + w^{i-}_j + w^{i+}_j)
            -
             p_j^{i+1}  w^{(i+1)-}_j
       \\ & = 
            p^i_j
            - ( C_j^{i-1} - \bar C_{j-1}^{i-1} )  \bar \gamma^{(i-1)+}_j
            +
            ( C_j^{i} - \bar C_{j-1}^{i} ) (\bar \gamma^{i-}_j + \bar \gamma^{i+}_j)
            -
            ( C_j^{i+1} - \bar C_{j-1}^{i+1} ) \bar \gamma^{(i+1)-}_j
        \ 
        \\ & \supstack{(*)} =\bar   p_{j-1} \ .
    \end{align*}
    The last equation~$(*)$ follows from
    \begin{align*}
         C_{j}^{i-1}  \bar \gamma^{(i-1)+}_j
                     -
                      C_{j}^{i} (\bar \gamma^{i-}_j + \bar \gamma^{i+}_j)
                      +
                       C_{j}^{i+1} \bar \gamma^{(i+1)-}_j   
       & =
        \Big(  C_{j}^{i+1} \bar \gamma^{(i+1)-}_j  - C_{j}^{i}  \bar \gamma^{i+}_j \Big)
        - \Big( C_{j}^{i} \bar \gamma^{i-}_j  -  C_{j}^{i-1}  \bar \gamma^{(i-1)+}_j \Big)
       \\& =
       dC^{i+1}_j
        - dC^{i}_j \\&=  p^i_j \ .
    \end{align*}
    For $i=1$ recall that $p^1_j = \frac{ C^2_j-C^1_j }{ K^2_j - K^1_j }$ with $C^1_j=1-K_j^1$
    hence
   $$
            ( Q_{j|j-1}^{-1} \cdot p_j  )^1
            = p_j^1 - \left( C^2_j - \bar C^2_{j-1} \right)\frac1{ K^2_j-K^1_j } 
            = \frac{ C^2_j-C^1_j }{ K^2_j - K^1_j }- \left( C^2_j -\bar  C^2_{j-1} \right)\frac1{ K^2_j-K^1_j } = \bar p^1_{j-1} \ .
    $$
    Similarly, for $i=N$ we have $p^N_j = \frac{ C^{N+1}_j-C^N_j }{ K_{N+1} - K_1 }$
    with $C^{N+1}_j=0$ such that
    $$
        ( Q_{j|j-1}^{-1} \cdot p_j  )^{N_j}
        = p^{N_j}_j - ( C^{N_j}_j -\bar  C^{N_j}_{j-1} ) \frac1{ K^{N_j+1}_j - K^{N_j}_j }
        = \frac{ C^{{N_j}+1}_j-C^{N_j}_j }{ K^{{N_j}+1}_j - K^{N_j}_j } - ( C^{N_j}_j - \bar C^{N_j}_{j-1} ) \frac1{ K^{{N_j}+1}_j - K^{N_j}_j }
        = \bar p^{N_j}_{j-1} \ .
    $$
\end{proof}

   \bibliography{sorted}
   \bibliographystyle{alpha}

\end{document}

%% file: mathcommands.tex
\usepackage{graphicx}

\usepackage[T1]{fontenc}    
\usepackage{booktabs}       

\usepackage{amsfonts,amssymb,bbm}
\usepackage{amsmath}
\usepackage{enumitem}
\usepackage{graphicx}
\usepackage{xcolor}
\usepackage{adjustbox}
\usepackage[UKenglish]{isodate}
\usepackage{graphicx}
\usepackage{silence}
\usepackage[format=default,font=small,labelfont=it]{caption}
\usepackage{subcaption}
\usepackage{hyperref} 
    \hypersetup{
    	colorlinks = true,
    	linkcolor = blue,
    	anchorcolor = blue,
    	citecolor = blue,
    	filecolor = blue,
    	urlcolor = blue
    }

\usepackage{float}
\usepackage{enumitem}
\usepackage{multirow}
\usepackage{fancyvrb}
\usepackage{rotating}
    \usepackage{tikz}
    \usepackage{tikz-cd} 
    \pgfdeclarelayer{edgelayer}
    \pgfdeclarelayer{nodelayer}
    \pgfsetlayers{edgelayer,nodelayer,main}
    \tikzstyle{new style 0}=[fill={rgb,255: red,255; green,94; blue,247}, draw=black, shape=circle]
    \tikzstyle{pointy}=[fill=white, draw=black, shape=circle]
    \tikzstyle{pointy}=[->]
\usepackage{fontawesome}

\usepackage{algpseudocode}
\usepackage[linesnumbered,lined,boxed,commentsnumbered,ruled,longend]{algorithm2e}

\SetCommentSty{mycommfont}

\usepackage{lscape}

\makeatletter
\newcommand{\pushright}[1]{\ifmeasuring@#1\else\omit\hfill$\displaystyle#1$\fi\ignorespaces}
\newcommand{\pushleft}[1]{\ifmeasuring@#1\else\omit$\displaystyle#1$\hfill\fi\ignorespaces}
\makeatother

\newcommand{\1}{\mathbbm{1}}

\renewcommand{\phi}{\varphi}

\newcommand{\E}{\mathbb{E}}

\renewcommand{\P}{\mathbb{P}}
\newcommand{\Q}{\mathbb{Q}}
\newcommand{\R}{\mathbb{R}}

\newcommand{\K}{\mathbb{K}}

\usepackage{etoolbox}
\newcommand{\addQEDstyle}[2]{\AtBeginEnvironment{#1}{\pushQED{\qed}\renewcommand{\qedsymbol}{#2}}\AtEndEnvironment{#1}{\popQED}}
\addQEDstyle{remark}{$\triangleleft$}
\addQEDstyle{example}{$\triangleleft$}

\NewDocumentCommand{\luca}{mo}{
    \IfValueF{#2}{
                        {{\scriptsize
                            \textcolor{green}{ 
                            \textbf{L:}
                            \textit{{#1}}
                            }
                        }}
        }
    \IfValueT{#2}{
                        \marginnote{{\scriptsize
                            \textcolor{green}{ 
                            \textbf{L:}
                            \textit{{#1}}
                            }
                        }}
        }
                    }
\NewDocumentCommand{\giulia}{mo}{
    \IfValueF{#2}{
                        {{\scriptsize
                            \textcolor{red}{ 
                            \textbf{GL:}
                            \textit{{#1}}
                            }
                        }}
        }
    \IfValueT{#2}{
                        \marginnote{{\scriptsize
                            \textcolor{red}{ 
                            \textbf{GL:}
                            \textit{{#1}}
                            }
                        }}
        }
}
\NewDocumentCommand{\anastasis}{mo}{
    \IfValueF{#2}{
                        {{\scriptsize
                            \textcolor{violet}{ 
                            \textbf{A:}
                            \textit{{#1}}
                            }
                        }}
        }
    \IfValueT{#2}{
                        \marginnote{{\scriptsize
                            \textcolor{violet}{ 
                            \textbf{A:}
                            \textit{{#1}}
                            }
                        }}
        }
                    }
                    
\NewDocumentCommand{\cody}{mo}{
    \IfValueF{#2}{
                        {{\scriptsize
                            \textcolor{orange}{ 
                            \textbf{A:}
                            \textit{{#1}}
                            }
                        }}
        }
    \IfValueT{#2}{
                        \marginnote{{\scriptsize
                            \textcolor{orange}{ 
                            \textbf{A:}
                            \textit{{#1}}
                            }
                        }}
        }
                    }

\NewDocumentCommand{\yannick}{mo}{
    \IfValueF{#2}{
                        {{\scriptsize
                            \textcolor{cyan}{ 
                            \textbf{Y:}
                            \textit{{#1}}
                            }
                        }}
        }
    \IfValueT{#2}{
                        \marginnote{{\scriptsize
                            \textcolor{cyan}{ 
                            \textbf{Y:}
                            \textit{{#1}}
                            }
                        }}
        }
                    } 

\definecolor{darkgreen}{rgb}{0.0, 0.2, 0.13}
\NewDocumentCommand{\xuwei}{mo}{
    \IfValueF{#2}{
                        {{\scriptsize
                            \textcolor{darkgreen}{ 
                            \textbf{X:}
                            \textit{{#1}}
                            }
                        }}
        }
    \IfValueT{#2}{
                        \marginnote{{\scriptsize
                            \textcolor{darkgreen}{ 
                            \textbf{X:}
                            \textit{{#1}}
                            }
                        }}
        }
                    }

\usepackage{appendix}

\usepackage{cleveref}
\newcounter{termcounter}
\renewcommand{\thetermcounter}{\Roman{termcounter}}
\crefname{term}{term}{terms}
\creflabelformat{term}{#2\textup{(#1)}#3}

\makeatletter
\def\term{\@ifnextchar[\term@optarg\term@noarg}
\def\term@optarg[#1]#2{%
  \textup{#1}%
  \def\@currentlabel{#1}%
  \def\cref@currentlabel{[][2147483647][]#1}%
  \cref@label[term]{#2}}
\def\term@noarg#1{%
  \refstepcounter{termcounter}%
  \textup{(\thetermcounter)}%
  \cref@label[term]{#1}}
\makeatother

\crefname{lemma}{lemma}{lemmata}
\Crefname{lemma}{Lemma}{Lemmata}
\crefname{assumption}{assumption}{assumptions}
\Crefname{assumption}{Assumption}{Assumptions}
\crefname{example}{Example}{Examples}
\crefname{proposition}{Proposition}{Proposition}
\crefname{algorithm}{algorithm}{algorithms}
\Crefname{algorithm}{Algorithm}{Algorithms}

\crefname{appendix}{appendix}{appendices}
\Crefname{appendix}{Appendix}{Appendices}

\usepackage{comment}

\usepackage{microtype}

%% file: comments.tex
\definecolor{bleudefrance}{rgb}{0.19, 0.55, 0.91}
\definecolor{britishracinggreen}{rgb}{0.0, 0.26, 0.15}
\definecolor{darkcyan}{rgb}{0.0, 0.55, 0.55}
\definecolor{MidnightBlue}{RGB}{25,25,112}
\definecolor{MidnightBlueComplementingGreen}{RGB}{25,112,25}
\definecolor{MidnightBlueComplementingPurple}{RGB}{112,25,112}
\definecolor{MidnightBlueComplementingRed}{RGB}{112,25,69}
\definecolor{darkmidnightblue}{rgb}{0.0, 0.2, 0.4}

\definecolor{WowColor}{rgb}{.75,0,.75}
\definecolor{MildlyAlarming}{rgb}{0.85,0.25,0.1}
\definecolor{SubtleColor}{rgb}{0,0,.50}
\definecolor{antiquefuchsia}{rgb}{0.57, 0.36, 0.51}
\definecolor{fashionfuchsia}{rgb}{0.96, 0.0, 0.63}
\definecolor{jade}{rgb}{0.0, 0.66, 0.42}
\definecolor{caribbeangreen}{rgb}{0.0, 0.8, 0.6}
\definecolor{aquamarine}{rgb}{0.5, 0.8, 0.85}
\definecolor{lightseagreen}{rgb}{0.13, 0.7, 0.67}
\definecolor{darkgreen}{rgb}{0.0, 0.2, 0.13}
\definecolor{darkspringgreen}{rgb}{0.09, 0.45, 0.27}
\definecolor{attentioncolor}{RGB}{152,90,81}
\definecolor{burgred}{RGB}{40,3,22}
\definecolor{AnnieGreen}{RGB}{17,123,92}
\definecolor{Turquoise}{RGB}{64,224,208}

\definecolor{darkjade}{RGB}{0,122,84}
\definecolor{Window1}{RGB}{92,150,31}%
    \definecolor{Window1dark}{RGB}{41,67,13}%
\definecolor{Window2}{RGB}{255,168,28}
    \definecolor{Window2dark}{RGB}{114,75,12}
\definecolor{Window3}{RGB}{255,96,33}
    \definecolor{Window3dark}{RGB}{97,36,12}
\definecolor{InputColor}{RGB}{20,255,177}
    \definecolor{InputColorlight}{RGB}{222,237,229}

\definecolor{RedAlizarin}{rgb}{0.82, 0.1, 0.26}
\usepackage{soul}

\usepackage[colorinlistoftodos]{todonotes}

\NewDocumentCommand{\Takashi}{mo}{
    \IfValueF{#2}{
                        {{
                            \textcolor{magenta}{ 
                            [\textbf{Takashi:}
                            \textit{{#1}}]
                            }
                        }}
        }
    \IfValueT{#2}{
                        \marginnote{{\scriptsize
                            \textcolor{magenta}{ 
                            \textbf{T:}
                            \textit{{#1}}
                            }
                        }}
        }
                    }

\NewDocumentCommand{\AK}{mo}{
    \IfValueF{#2}{\hfill\\
                        {{
                            \textcolor{darkmidnightblue}{ 
                            \textbf{AK:}
                            \textit{{#1}}
                            }
                        }}
        }
    \IfValueT{#2}{
                        \marginnote{{\scriptsize
                            \textcolor{darkmidnightblue}{ 
                            \textbf{AK:}
                            \textit{{#1}}
                            }
                        }}
        }
                    }

\NewDocumentCommand{\yan}{mo}{
    \IfValueF{#2}{\hfill\\
                        {{
                            \textcolor{cyan}{ 
                            \textbf{YL:}
                            \textit{{#1}}
                            }
                        }}
        }
    \IfValueT{#2}{
                        \marginnote{{\scriptsize
                            \textcolor{cyan}{ 
                            \textbf{YL:}
                            \textit{{#1}}
                            }
                        }}
        }
                    }

\NewDocumentCommand{\RS}{mo}{
    \IfValueF{#2}{\hfill\\
                        {{
                            \textcolor{WowColor}{ 
                            \textbf{RS:}
                            \textit{{#1}}
                            }
                        }}
        }
    \IfValueT{#2}{
                        \marginnote{{\scriptsize
                            \textcolor{WowColor}{ 
                            \textbf{RS:}
                            \textit{{#1}}
                            }
                        }}
        }
}

\NewDocumentCommand{\HB}{mo}{
    \IfValueF{#2}{\hfill\\
                        {{
                            \textcolor{violet}{ 
                            \textbf{HB:}
                            \textit{{#1}}
                            }
                        }}
        }
    \IfValueT{#2}{
                        \marginnote{{\scriptsize
                            \textcolor{violet}{ 
                            \textbf{HB:}
                            \textit{{#1}}
                            }
                        }}
        }
                    }

\NewDocumentCommand{\BH}{mo}{
    \IfValueF{#2}{\hfill\\
                        {{
                            \textcolor{bleudefrance}{ 
                            \textbf{BH:}
                            \textit{{#1}}
                            }
                        }}
        }
    \IfValueT{#2}{
                        \marginnote{{\scriptsize
                            \textcolor{bleudefrance}{ 
                            \textbf{BH:}
                            \textit{{#1}}
                            }
                        }}
        }
                    }
